\documentclass[
a4paper,
fleqn,
]{article}

\usepackage
[
]
{geometry}

\usepackage{setspace}
\usepackage{Settings}

\newcommand{\FBV}{\mathop{\bigvee}
\limits
}

\newcommand{\FBW}{\mathop{\bigwedge}
\limits
}

\newcommand{\FBC}{\mathop{\bigcup}
\limits
}

\usepackage{lineno}

\title{Completeness of two fragments of a logic \\ for conditional strategic reasoning}
\author{
Yinfeng Li${}^{1,2}$ and Fengkui Ju${}^{3,4}$\footnote{Corresponding author} \vspace{5pt} \\
{\small {$^1$IRIT-CNRS, University of Toulouse, France}} \\
{\small {$^2$\href{mailto:yinfeng.li@irit.fr}{yinfeng.li@irit.fr}}} \vspace{2.5pt} \\
{\small {$^3$School of Philosophy, Beijing Normal University, China}} \\
{\small {$^4$\href{mailto:fengkui.ju@bnu.edu.cn}{fengkui.ju@bnu.edu.cn}}}
}
\date{}

\begin{document}

\setlength{\parskip}{0.5em}

%\linenumbers

\maketitle

%%%%%%%%%%%%%%%%%%%%%%%%%%%%%%%%%%%%%
%%%%%%%%%%%%%%%%%%%%%%%%%%%%%%%%%%%%%
\begin{abstract}

\noindent Classical logics for strategic reasoning, such as Coalition Logic and Alternating-time Temporal Logic, formalize absolute strategic reasoning about the unconditional strategic abilities of agents to achieve their goals.
Goranko and Ju in \cite{goranko_towards_2019,goranko_logic_2022} introduced a Logic for Conditional Strategic Reasoning ($\FCSR$). However, its completeness is still an open problem.
$\FCSR$ has three featured operators, and one of them has the following reading: For some action of A that guarantees the achievement of her goal, B has an action to guarantee the achievement of his goal. This operator makes good sense when A is cooperating with B. The logic about this operator is called Logic for Cooperating Conditional Strategic Reasoning ($\FCCSR$). In this paper, we prove the completeness of two fragments of $\FCCSR$: the liability fragment and the ability fragment.
The key ingredients of our proof approach include standard disjunctions, the validity-reduction condition of standard disjunctions, abstract game forms and their realization, and the derivability-reduction condition of standard disjunctions. The approach has good potential to be applied to the completeness of $\FCSR$ and other strategic logics.

\medskip


\end{abstract}

% This section needs no polishing.

%%%%%%%%%%%%%%%%%%%%%%%%
%%%%%%%%%%%%%%%%%%%%%%%%
\section{Introduction}
\label{section:Introduction}

%%%%%%%%%%%%%%%%%%%%%%%%
%%%%%%%%%%%%%%%%%%%%%%%%
\subsection{Logic for Conditional Strategic Reasoning $\FCSR$}

Logical research on strategic reasoning started with Coalition Logic $\FCL$ (\cite{pauly_logic_2001,pauly_modal_2002}) and Alternating-time Temporal Logic $\FATL$ (\cite{alur_alternating-time_2002}), which is a temporal extension of $\FCL$. The featured operator of $\FCL$ and $\FATL$ is $\Fclo{\FAA} \phi$, indicating \emph{some available joint action of the coalition $\FAA$ guarantees $\phi$}. There has been much following work, including Alternating-time Temporal Epistemic Logic $\FATEL$ (\cite{hoek_cooperation_2004}) and Strategy Logic $\FSL$ \cite{mogavero_reasoning_2014} (cf. \cite{mogavero_reasoning_2017}). We refer to \cite{benthem_models_2015} and \cite{agotnes_knowledge_2015} for overviews of the area.

Most of the present logical systems in the area assume arbitrary behavior of the agents outside of the proponent coalitions. Thus, these logics formalize \emph{absolute} strategic reasoning about the unconditional strategic abilities of agents to achieve their goals.

Usually, agents have their own goals and act in pursuit of their fulfillment rather than just to prevent the proponents from achieving their goals. This calls for more refined strategic reasoning, \emph{conditional} on agents' knowledge of opponents' goals and possible actions to achieve them.

Goranko and Ju \cite{goranko_towards_2019,goranko_logic_2022} presented a logic $\FCSR$ (Logic for Local Conditional Strategic Reasoning) for reasoning of the type: \textit{For some/every joint action of the coalition $\FAA$ that guarantees its goal, the coalition $\FBB$ has a joint action to guarantee its goal.} \footnote{The notation for this logic is $\mathsf{ConStR}$ in \cite{goranko_towards_2019,goranko_logic_2022}. In this paper, we use a simpler one.} 
The logic $\FCSR$ extends Coalition Logic with the following three featured operators:
\begin{itemize}
\item $\Fchance{\FAA}{\phi}{\FBB}{\psi}$: \emph{$\coA$ has an available joint action $\sigma_\FAA$ such that (1) it guarantees $\phi$ and (2) $\FBB$ has an available joint action $\sigma_\FBB$ that guarantees $\psi$.}
\item $\dere{\FAA}{\FBB}{\phi}{\psi}$: \emph{for every available joint action $\sigma_\FAA$ of $\FAA$ that guarantees $\phi$, $\FBB$ has an available joint action $\sigma_\FBB$ that guarantees $\psi$.}
\item $\dedic{\FAA}{\FBB}{\phi}{\psi}$: \emph{$\FBB$ has an available joint action $\sigma_{\FBB}$ such that for every available joint action $\sigma_\FAA$ of $\FAA$ that guarantees $\phi$, $\sigma_{\FBB}$ guarantees $\psi$.}
\end{itemize}

Note that for these operators, $\FAA$ and $\FBB$ perform their actions simultaneously. The operator $\Fchance{\FAA}{\phi}{\FBB}{\psi}$ makes good sense when $\FAA$ is \emph{cooperating} with $\FBB$, which is why it has the subscript ``c''.
The operators $\dere{\FAA}{\phi}{\FBB}{\psi}$ and $\dedic{\FAA}{\FBB}{\phi}{\psi}$ have close connection to \emph{$\beta$ effecitivity} and \emph{$\alpha$ effectivity} in game theory, respectively, which is why they have subscripts ``$\beta$'' and ``$\alpha$''. 

We refer to \cite{goranko_towards_2019,goranko_logic_2022} for some real scenarios where these operators can be used. In what follows, we give an example to illustrate them.

%%%%%%%%%%%%%%%%%%%%%%%%
%%%%%%%%%%%%%%%%%%%%%%%%
\begin{example}
\label{example:card}

Alice and Bob are playing a simple one-round game. Alice and Bob have three cards each. On every card, there are two natural numbers, called an F-number and an S-number, respectively. Here, ``F'' is for ``first'' and ``S'' is for ``second''. Both Alice and Bob will show a card at the same time. There are two kinds of winning: F-winning and S-winning, respectively determined by which F-number is greater and which S-number is greater.
\begin{itemize}

%%%%
\item

Suppose Alice's three cards are $(5,3)$, $(4,8)$, $(1,8)$, and Bob' three cards are $(4,4)$, $(3,2)$, $(2,2)$. Then, the following holds:
\textbf{Alice has a card such that
(1) it guarantees Alice to F-win, and
(2) Bob has a card that guarantees Bob to S-win.}

\emph{Why does the statement hold? Given that Alice is going to show $(5,3)$, (1) Alice will F-win for sure, and (2) given Bob is going to show $(4,4)$, Bob will S-win for sure.
Note that no card lets Bob S-win for sure.
This statement can be expressed by the operator $\Fchance{\FAA}{\phi}{\FBB}{\psi}$.}

%%%%
\item

Suppose Alice's three cards are $(5,3)$, $(7,2)$, $(1,8)$, and Bob' three cards are $(4,4)$, $(3,3)$, $(2,2)$. Then, the following holds:
\textbf{For every card of Alice, if it guarantees Alice to F-win, then Bob has a card such that it guarantees Bob to S-win.}

\emph{%
Why does the statement hold? Alice has two cards to F-win for sure: $(5,3)$ and $(7,2)$. Given that Alice is going to show $(5,3)$, the card $(4,4)$ will let Bob S-win for sure. Given that Alice is going to show $(7,2)$, the card $(3,3)$ will let Bob S-win for sure.
Again, note that no card lets Bob S-win for sure.
This statement can be expressed by the operator $\dere{\FAA}{\FBB}{\phi}{\psi}$.
}

%%%%
\item

Suppose Alice's three cards are $(6,3)$, $(5,8)$, $(1,9)$, and Bob' three cards are $(4,6)$, $(3,4)$, $(2,9)$. Then, the following holds:
\textbf{Bob has a card such that, for every card of Alice, if it guarantees Alice to F-win, then Bob's card guarantees Bob to S-win.}

\emph{%
Why does the statement hold? Alice has two cards to F-win for sure: $(6,3)$ and $(5,8)$.
Suppose Bob is going to show $(2,9)$. Then, 
given that Alice is going to show $(6,3)$, Bob will S-win for sure, and given that Alice is going to show $(5,8)$, Bob will S-win for sure.
Again, note that no card lets Bob S-win for sure.
This statement can be expressed by the operator $\dedic{\FAA}{\FBB}{\phi}{\psi}$.
}

\end{itemize}

\end{example}

It is shown in \cite{goranko_towards_2019,goranko_logic_2022} that none of the operators $\Fchance{\FAA}{\phi}{\FBB}{\psi}$, $\dere{\FAA}{\phi}{\FBB}{\psi}$, and $\dedic{\FAA}{\phi}{\FBB}{\psi}$ can be expressed in Coalition Logic and Alternating-time Temporal Logic.
Generally speaking, the reason is that these operators express the dependence of strategies of $\FBB$ on strategies of $\FAA$, which Coalition Logic and Alternating-time Temporal Logic cannot express.

\paragraph{Some related work with $\FCSR$}

In \cite{goranko_socially_2018}, Goranko and Enqvist introduced Group Protecting Coalition Logic $\mathsf{GPCL}$, which focuses on the synchronization of different coalitions to achieve their respective goals. The featured formula of $\mathsf{GPCL}$ is $\Fclo{\FAA_1 \leadto \phi_1, \dots, \FAA_n \leadto \phi_n}$, meaning \emph{there is an action profile such that its restriction to $\FAA_1$ guarantees $\phi_1$, \dots, and its restriction to $\FAA_n$ guarantees $\phi_n$.} The operator $\Fchance{\FAA}{\phi}{\FBB}{\psi}$ of $\FCSR$ is definable in $\mathsf{GPCL}$. Enqvist and Goranko \cite{enqvist_temporal_2022} studied Temporal Logic of Coalitional Goal Assignments $\mathsf{TLCGA}$, which is a temporal extension of $\mathsf{GPCL}$.
In \cite{naumov_intelligence_2021}, Naumov and Yuan presented a logic for reasoning about the powers of coalitions based on their information about the moves of other coalitions. The featured operator of the logic is $[\FBB]_\FAA \psi$, meaning
\emph{for every joint action of $\FAA$, given that $\FAA$ is going to perform it and $\FBB$ distributively knows this, $\FBB$ has a joint action to guarantee $\psi$}.
Without considering knowledge, $[\FBB]_\FAA \psi$ is definable by $\dere{\FAA}{\FBB}{\top}{\psi}$ in $\FCSR$.

%%%%%%%%%%%%%%%%%%%%%%%%
%%%%%%%%%%%%%%%%%%%%%%%%
\subsection{Our work}

The completeness of $\FCSR$ is still open. Intuitively, this is challenging work: the meaning of the three operators of $\FCSR$ involves complex iterations of quantifiers. In this paper, we want to do some useful work towards showing the completeness of $\FCSR$.

We call the sub-logic of $\FCSR$ with the operator $\Fchance{\FAA}{\phi}{\FBB}{\psi}$ Logic for Cooperating Conditional Strategic Reasoning $\FCCSR$. In this work, we show the completeness of two fragments of $\FCCSR$: the ability fragment and the liability fragment. The former is generated by $\Fchance{\FAA}{\phi}{\FBB}{\psi}$ without negation, and the latter is generated by the dual of $\Fchance{\FAA}{\phi}{\FBB}{\psi}$ without negation.
%
%Our motivation to show the completeness of the two fragments is primarily technical.

Generally speaking, our proof approach is as follows.
First, we show a normal form lemma, by which every formula can be transformed to a conjunction of some so-called standard disjunctions.
Second, we show a downward validity lemma, which reduces the validity of a standard disjunction to the validity of its subformulas with lower complexity.
Third, we show an upward derivability lemma, which reduces the derivability of a formula to the derivability of its subformulas with lower complexity.
Fourth, we show the completeness by induction.
This approach has some connection to the approach of Goranko and van Drimmelen \cite{goranko_complete_2006}, who showed the completeness of $\FATL$.

The remaining part of the paper is structured as follows. 
Section \ref{section:Preliminaries} provides some preliminaries on concurrent game models, the Coalition Logic $\FCL$, the liability fragment and the ability fragment of $\FCL$, the logic $\FCCSR$ for cooperating conditional strategic reasoning, and the liability fragment and the ability fragment of $\FCCSR$.
Section \ref{section:Our completeness proof strategy} presents our completeness proof approach, and Section \ref{section:A model construction technique} introduces regular abstract game forms and their realization, which will be used later.
Section \ref{section:Completeness of the liability fragment of FCL} shows the completeness of the liability fragment of $\FCL$, and Section \ref{section:Completeness of the ability fragment of FCL} shows the completeness of the liability fragment of $\FCL$. The proofs in the two sections share some similarities with the proofs in the following two sections. The purpose of the two sections is to offer readers some preparations for the proofs in the following two sections.
Section \ref{section:Completeness of the liability fragment of FCCSR} shows the completeness of the liability fragment of $\FCCSR$, and Section \ref{section:Completeness of the ability fragment of FCCSR} shows the completeness of the ability fragment of $\FCCSR$.
We end with brief concluding remarks in Section \ref{section:Concluding remarks}.

%%%%%%%%%%%%%%%%%%%%%%%%
%%%%%%%%%%%%%%%%%%%%%%%%
\section{Preliminaries}
\label{section:Preliminaries}

In this section, we present models, languages, and semantics of Coalition Logic $\FCL$ and Logic for Cooperating Conditional Strategic Reasoning $\FCCSR$. In addition, we also specify two fragments of $\FCL$ and two fragments of $\FCCSR$.

%%%%%%%%%%%%%%%%%%%%%%%%
%%%%%%%%%%%%%%%%%%%%%%%%
\subsection{Concurrent game models}

In the sequel, we assume a nonempty finite set of \emph{agents} $\FAG$ and a countable set of atomic propositions $\FAP$.

Let $\FAC$ be a nonempty set of \emph{actions}. What follows are some auxiliary notions and notations.
\begin{itemize}

%%%%
\item

Each (possibly empty) subset $\FAA$ of $\FAG$ is called a \Fdefs{coalition}. $\FAG$ is called the \Fdefs{grand coalition}. In the sequel, we will often write $a$ instead of $\{a\}$ when no confusion arises.

%%%%
\item

For every coalition $\FAA$, a function $\sigma_\FAA: \FAA \rightarrow \FAC$ is called a \Fdefs{joint action} of $\FAA$. Note that $\emptyset$ is the only joint action of the empty coalition: the empty set is the only function from the empty coalition to $\FAC$. A joint action of $\FAG$ is called an \Fdefs{action profile}.

For every $\FAA \subseteq \FAG$, we define $\FJA_\FAA = \{\sigma_\FAA \mid \sigma_\FAA: \FAA \rightarrow \FAC\}$, which is the set of joint actions of $\FAA$.
We use $\FJA$ to indicate $\bigcup \{\FJA_\FAA \mid \FAA \subseteq \FAG\}$.

%%%%
\item

Let $\FAA$ be a coalition and $\FBB \subseteq \FAA$. Let $\sigma_\FAA$ be a joint action of $\FAA$. We use $\sigma_\FAA|_\FBB$ to denote the subset of $\sigma_\FAA$ which is a joint action of $\FBB$, called the \Fdefs{restriction} of $\sigma_\FAA$ to $\FBB$. Respectively, $\sigma_\FAA$ is called an \Fdefs{extension} of $\sigma_\FAA|_\FBB$. Given a set of joint actions $\Sigma_\FAA$ of $\FAA$, we define $\Sigma_\FAA |_\FBB := \{\sigma_\FAA|_\FBB \mid \sigma_\FAA \in \Sigma_\FAA\}$.

%%%%
\item

Let $\FAA$ and $\FBB$ be two disjoint coalitions, $\Sigma_\FAA$ be a set of joint actions of $\FAA$, and $\Sigma_\FBB$ be a set of joint actions of $\FBB$. We define $\Sigma_\FAA \oplus \Sigma_\FBB$ as $\{\sigma_\FAA \cup \sigma_\FBB \mid \sigma_\FAA \in \Sigma_\FAA \mbox { and } \sigma_\FBB \in \Sigma_\FBB\}$, which is a set of joint actions of $\FAA \cup \FBB$.

Further, let $\{\FAA_i \mid i \in I\}$ be a family of pairwise disjoint coalitions for some (possibly empty) index set $I$. For every $i \in I$, let $\Sigma_{\FAA_i}$ be a set of joint actions of $\FAA_i$.
Define:
\[
\bigoplus \{\Sigma_{\FAA_i} \mid i \in I\} = \{\FBC \Delta \mid \Delta \subseteq \FBC \{\Sigma_{\FAA_i} \mid i \in I\} \text{ and for every } i \in I, \Delta \cap \Sigma_{\FAA_i} \text{ is a singleton}\}
\]
which is a set of joint actions of the coalition $\FBC \{\FAA_i \mid i \in I\}$.

\end{itemize}

%%%%%%%%%%%%%%%%%%%%%%%%
%%%%%%%%%%%%%%%%%%%%%%%%
\begin{definition}[Concurrent game models]

A \defstyle{concurrent game model} is a tuple $\MM = (\FST, \FAC,$ $ \Faja, \Fout, \Flab)$, where:
\begin{itemize}

%%%%
\item

$\FST$ is a nonempty set of states.

%%%%
\item

$\FAC$ is a nonempty set of actions.

%%%%
\item

$\Faja: \FST \times \mathcal{P}(\FAG) \rightarrow \mathcal{P}(\FJA)$ is a function, called an \Fdefs{available joint action function}, such that for every $s \in \FST$ and $\FAA \subseteq \FAG$:
\begin{itemize}

%%%%
\item

$\Faja (s,\FAA)$ is a nonempty subset of $\FJA_\FAA$;

%%%%
\item

$\Faja (s,\FAA) = \bigoplus \{\Faja(s,\{a\}) \mid a \in \FAA\}$, where $\FAA \neq \emptyset$.

\end{itemize}

\emph{Here $\Faja (s,\FAA)$ is the set of all available joint actions of $\FAA$ at $s$.} \emph{Note $\Faja (s,\emptyset) = \{\emptyset\}$.}

%%%%
\item

$\Fout: \FST \times \FJA \rightarrow \mathcal{P}(\FST)$ is a function, called an \Fdefs{outcome state function}, such that for every $s \in \FST$, $\FAA \subseteq \FAG$, and $\sigma_\FAA \in \FJA_\FAA$:
\begin{itemize}

%%%%
\item

if $\FAA = \FAG$ and $\sigma_\FAA \in \Faja (s,\FAA)$, then $\Fout (s, \sigma_\FAA)$ is a singleton;

%%%%
\item

if $\FAA = \FAG$ and $\sigma_\FAA \notin \Faja (s,\FAA)$, then $\Fout (s, \sigma_\FAA) = \emptyset$;

%%%%
\item

if $\FAA \neq \FAG$ and $\sigma_\FAA \in \Faja (s,\FAA)$, then $\Fout (s, \sigma_\FAA) = \FBC \{\Fout (s, \sigma_\FAG) \mid \sigma_\FAG \in \Faja (s,\FAG) \text{ and}$ $\sigma_\FAA \subseteq \sigma_\FAG\}$;

%%%%
\item

if $\FAA \neq \FAG$ and $\sigma_\FAA \notin \Faja (s,\FAA)$, then $\Fout (s, \sigma_\FAA) = \emptyset$.

\end{itemize}

\emph{Here $\Fout (s, \sigma_\FAA)$ is the set of outcome states of performing $\sigma_\FAA$ at $s$.} \emph{Note $\Fout (s,\emptyset) = \FBC \{\Fout (s, \sigma_\FAG) \mid \sigma_\FAG \in \Faja (s,\FAG)\}$.}

%%%%
\item

$\Flab: \FST \rightarrow \mathcal{P}(\FAP)$ is a labeling function.

\end{itemize}

\end{definition}

For any model $\MM$ and state $s$ of $\MM$, $(\MM,s)$ is called a \Fdefs{pointed model}.

%%%%%%%%%%%%%%%%%%%%%%%%
%%%%%%%%%%%%%%%%%%%%%%%%
\subsection{Coalition Logic $\FCL$ and its two fragments}

%%%%%%%%%%%%%%%%%%%%%%%%
%%%%%%%%%%%%%%%%%%%%%%%%
\begin{definition}[The language $\Phi_{\FCL}$]

The language $\Phi_{\FCL}$ is defined as follows:
\[
\phi ::=\top \mid p \mid \neg \phi \mid (\phi \wedge \phi) \mid \Fclo{\FAA} \phi
\]

\end{definition}

The propositional connectives $\bot, \lor, \rightarrow$, and $\leftrightarrow$ are defined as usual.
Define the dual $\Fclod{\FAA} \phi$ of $\Fclo{\FAA} \phi$ as $\neg \Fclo{\FAA} \neg \phi$, whose intuitive reading is: \emph{every available joint action of $\FAA$ is liable to $\phi$}.

Here, the notations of $\Fclo{\cdot} \phi$ and $\Fclod{\cdot}$ of the two modalities differ from their notations in the literature. We do this to indicate the iterations of quantifiers in their meaning. 

%%%%%%%%%%%%%%%%%%%%%%%%
%%%%%%%%%%%%%%%%%%%%%%%%
\begin{definition}[Semantics of $\Phi_{\FCL}$]
~

\begin{center}

\begin{tabular}{lll}
$\MM, s \Vdash \top$ & & \\
$\MM, s \Vdash \neg \phi$ & $\Leftrightarrow$ & \parbox[t]{27em}{not $\MM, s \Vdash \phi$} \\
$\MM, s \Vdash \phi \land \psi$ & $\Leftrightarrow$ & \parbox[t]{27em}{$\MM, s \Vdash \phi$ and $\MM, s \Vdash \psi$} \\
$\MM, s \Vdash \Fclo{\FAA} \phi$ & $\Leftrightarrow$ & \parbox[t]{27em}{$\FAA$ has a joint action $\sigma_\FAA$ in $\Faja (s, \FAA)$ such that $\MM, t \Vdash \phi$ for every $t \in \Fout (s, \sigma_\FAA)$}
\end{tabular}

\end{center}

\end{definition}

It can be verified:

\begin{center}
\begin{tabular}{lll}
$\MM, s \Vdash \Fclod{\FAA} \phi$ & $\Leftrightarrow$ & \parbox[t]{27em}{for all joint actions $\sigma_\FAA$ of $\FAA$ in $\Faja (s, \FAA)$, $\MM, t \Vdash \phi$ for some $t \in \Fout (s, \sigma_\FAA)$}
\end{tabular}
\end{center}

Here are some notations that will be used later:
\begin{itemize}
\item We use $\MM, S \Vdash \phi$ to state that for all $s \in S$, $\MM, s \Vdash \phi$, where $S$ is a set of states.
\item We use $\ja{\FAA} \leadto_s \phi$ to state that $\MM, t \Vdash \phi$ for all $t \in \Fout (s, \ja{\FAA})$.
The meaning of $\ja{\FAA} \leadto_s \Gamma$ is as expected, where $\Gamma$ is a set of formulas.
\item We use $\ja{\FAA} \Fena_s \phi$ to state that $\MM, t \Vdash \phi$ for some $t \in \Fout (s, \ja{\FAA})$.
\end{itemize}

\medskip

%%%%%%%%%%%%%%%%%%%%%%%%
%%%%%%%%%%%%%%%%%%%%%%%%
\begin{definition}[The liability fragment $\Phi_{\FCLn}$ of $\Phi_{\FCL}$]

The liability fragment $\Phi_{\FCLn}$ of $\Phi_{\FCL}$ is defined as follows:
\[
\phi ::= \top \mid \bot \mid p \mid \neg p \mid (\phi \wedge \phi) \mid (\phi \vee \phi) \mid \Fclod{\FAA} \phi
\]

\end{definition}

%%%%%%%%%%%%%%%%%%%%%%%%
%%%%%%%%%%%%%%%%%%%%%%%%
\begin{definition}[The ability fragment $\Phi_{\FCLp}$ of $\Phi_{\FCL}$]

The ability fragment $\Phi_{\FCLp}$ of $\Phi_{\FCL}$ is defined as follows:
\[
\phi ::= \top \mid \bot \mid p \mid \neg p \mid (\phi \wedge \phi) \mid (\phi \vee \phi) \mid \Fclo{\FAA}{\phi}
\]

\end{definition}

Here, ``LI'' and ``AB'' in ``$\FCLn$'' and ``$\FCLp$'' are for ``liability'' and ``ability'', respectively.

%%%%%%%%%%%%%%%%%%%%%%%%
%%%%%%%%%%%%%%%%%%%%%%%%
\subsection{Logic for Cooperating Conditional Strategic Reasoning $\FCCSR$ and its two fragments}

%%%%%%%%%%%%%%%%%%%%%%%%
%%%%%%%%%%%%%%%%%%%%%%%%
\begin{definition}[The language $\Phi_{\FCCSR}$]

The language $\Phi_{\FCCSR}$ is defined as follows:
\[
\phi ::= p \mid \top \mid \neg \phi \mid (\phi \land \phi) \mid \Fchance{\FAA}{\phi}{\FBB}{\phi}
\]

\end{definition}

Define the dual $\Fchanced{\FAA}{\phi}{\FBB}{\psi}$ of $\Fchance{\FAA}{\phi}{\FBB}{\psi}$ as $\neg \Fchance{\FAA}{\neg \phi}{\FBB}{\neg \psi}$, whose intuitive reading is: \emph{every availabe joint action of $\FAA$ that guarantees $\neg \phi$ also prevents $\FBB$ to guarantee $\neg \psi$}. An alternative reading of it is: \emph{for every available joint action of $\FAA$, either it is liable to $\phi$, or for every joint action of $\FBB$, it is liable to $\psi$.}

This operator formalizes the conditional reasoning scenario where the goals of $\FAA$ and $\FBB$ are conflicting, and whichever way $\FAA$ acts towards the guaranteed achievement of its goal, that would block $\FBB$ from acting to guarantee the achievement of its goal.

Fix a concurrent game model $\MM = (\FST, \FAC, \Faja, \Fout, \Flab)$. For all coalitions $\FAA$ and $\FBB$, and joint actions $\sigma_\FAA$ of $\FAA$ and $\sigma_\FBB$ of $\FBB$, define $\sigma_\FAA \uplus \sigma_\FBB = \sigma_\FAA \cup \{(b, \beta) \in \sigma_\FBB \mid b \in \FBB-\FAA \text{ and } \beta \in \FAC\}$, which is a joint action of $\FAA \cup \FBB$.

%%%%%%%%%%%%%%%%%%%%%%%%
%%%%%%%%%%%%%%%%%%%%%%%%
\begin{definition}[Semantics of $\Phi_{\FCCSR}$]
~
\begin{center}

\begin{tabular}{lll}
$\MM, s \Vdash \top$ & & \\
$\MM, s \Vdash \neg \phi$ & $\Leftrightarrow$ & \parbox[t]{27em}{not $\MM, s \Vdash \phi$} \\
$\MM, s \Vdash \phi \land \psi$ & $\Leftrightarrow$ & \parbox[t]{27em}{$\MM, s \Vdash \phi$ and $\MM, s \Vdash \psi$} \\
$\MM, s \Vdash \Fchance{\FAA}{\phi}{\FBB}{\psi}$ & $\Leftrightarrow$ & \parbox[t]{27em}
{
$\FAA$ has a joint action $\sigma_\FAA$ in $\Faja (s,\FAA)$ such that (1) $\MM, t \Vdash \phi$ for every $t \in \Fout (s, \sigma_\FAA)$, and (2) $\FBB$ has a joint action $\sigma_\FBB$ in $\Faja (s,\FBB)$ such that $\MM, t \Vdash \psi$ for every $t \in \Fout (s, \sigma_\FAA \uplus \sigma_\FBB)$
}
\end{tabular}
\end{center}

\end{definition}

It can be verified:

\begin{center}
\begin{tabular}{lll}
$\MM, s \Vdash \Fchanced{\FAA}{\phi}{\FBB}{\psi}$ & $\Leftrightarrow$ & \parbox[t]{27em}{for every joint action $\sigma_\FAA$ of $\FAA$ in $\Faja (s,\FAA)$, either (1) $\MM, t \Vdash \phi$ for some $t \in \Fout (s, \sigma_\FAA)$, or (2) for every joint action $\sigma_\FBB$ of $\FBB$ in $\Faja (s,\FBB)$, $\MM, t \Vdash \psi$ for some $t \in \Fout (s, \sigma_\FAA \uplus \sigma_\FBB)$}
\end{tabular}
\end{center}

We say that (1) a formula $\phi$ is \Fdefs{satisfiable} if for some pointed model $(\MM,s)$, $\MM, s \Vdash \phi$, (2) $\phi$ is \Fdefs{valid} if for every pointed model $(\MM,s)$, $\MM, s \Vdash \phi$, and (3) $\phi$ is \Fdefs{equivalent} to a formula $\phi'$ if $\phi \leftrightarrow \phi'$ is valid.

\medskip

The following fact can be easily shown, and we skip its proof:

\begin{fact}
~
\begin{enumerate}
\item $\MM, s \Vdash \Fchance{\FAA}{\phi}{\FBB}{\psi}$ if and only if $\FAA \cup \FBB$ has a joint action $\sigma_{\FAA \cup \FBB}$ in $\Faja (s,\FAA \cup \FBB)$ such that (1) $\MM, t \Vdash \phi$ for every $t \in \Fout (s, \sigma_{\FAA \cup \FBB} |_{\FAA})$, and (2) $\MM, t \Vdash \psi$ for every $t \in \Fout (s, \sigma_{\FAA \cup \FBB})$.
\item The following formulas are equivalent:
\begin{enumerate}
\item $\Fchance{\FAA}{\phi}{\FBB}{\psi}$
\item $\Fchance{\FAA}{\phi}{\FBB}{(\phi \land \psi)}$
\item $\Fchance{\FAA}{\phi}{\FBB-\FAA}{\psi}$
\item $\Fchance{\FAA}{\phi}{\FAA \cup \FBB}{\psi}$
\end{enumerate}
\item The operator $\Fclo{\FAA} \phi$ in $\FCL$ is definable as $\Fchance{\FAA}{\phi}{\emptyset}{\top}$.
\end{enumerate}

\end{fact}

\medskip

%%%%%%%%%%%%%%%%%%%%%%%%
%%%%%%%%%%%%%%%%%%%%%%%%
\begin{definition}[The liability fragment $\Phi_{\FCCSRn}$ of $\Phi_{\FCCSR}$]

The liability fragment $\Phi_{\FCCSRn}$ of $\Phi_{\FCCSR}$ is defined as follows:
\[
\phi ::= \top \mid \bot \mid p \mid \neg p \mid (\phi \wedge \phi) \mid (\phi \vee \phi) \mid \Fchanced{\FAA}{\phi}{\FBB}{\phi}
\]

\end{definition}

%%%%%%%%%%%%%%%%%%%%%%%%
%%%%%%%%%%%%%%%%%%%%%%%%
\begin{definition}[The ability fragment $\Phi_{\FCCSRp}$ of $\Phi_{\FCCSR}$]

The ability fragment $\Phi_{\FCCSRp}$ of $\Phi_{\FCCSR}$ is defined as follows:
\[
\phi ::= \top \mid \bot \mid p \mid \neg p \mid (\phi \wedge \phi) \mid (\phi \vee \phi) \mid \Fchance{\FAA}{\phi}{\FBB}{\phi}
\]

\end{definition}

\newcommand{\FSD}{\mathsf{SD}}

\newcommand{\FSC}{\mathsf{SC}}

%%%%%%%%%%%%%%%%%%%%%%%%
%%%%%%%%%%%%%%%%%%%%%%%%
\section{Our proof strategy}
\label{section:Our completeness proof strategy}

In this section, we present our approach to show the completeness of the following four fragments: the liability fragment $\FCLn$ of $\FCL$, the ability fragment $\FCLp$ of $\FCL$, the liability fragment $\FCCSRn$ of $\FCCSR$, and the ability fragment $\FCCSRp$ of $\FCCSR$.

We use $\FXL$ to indicate an arbitrary member among the four fragments. We assume the usual notion of \Fdefs{modal depth} of formulas of $\FXL$.
We will show the completeness of $\FXL$ in four steps.

First, we show a \Fdefs{normal form lemma} for $\FXL$: \emph{for every $\phi$ in $\Phi_\FXL$, there is $\phi'$ in $\Phi_\FXL$ such that: (1) $\models \phi$ if and only if $\models \phi'$, (2) $\vdash_\FXL \phi$ if and only if $\vdash_\FXL \phi'$, (3) $\phi'$ has the same modal depth with $\phi$, and (4) $\phi'$ is in a normal form $\FSD_0 \land \dots \land \FSD_n$, where every $\FSD_i$ is a so-called \Fdefs{standard disjunction} for $\FXL$.}

Second, we show a \Fdefs{downward validity lemma} for $\FXL$: \textit{for every standard disjunction $\FSD$ for $\FXL$, if $\models \FSD$, then a \Fdefs{validity-reduction condition} of $\FSD$ is met.} Here, the validity-reduction condition of $\FSD$ concerns the validity of some formulas with lower modal depth than $\FSD$.
This step is crucial for the whole proof.
 
Third, we show a \Fdefs{upward derivability lemma} for $\FXL$: \textit{for every standard disjunction $\FSD$ for $\FXL$, if a \Fdefs{derivability-reduction condition} of $\FSD$ is met, then $\vdash_\FXL \FSD$.} Here, the derivability-reduction condition of $\FSD$ is the result of replacing \emph{validity} in the validity-reduction condition of $\FSD$ by \emph{derivability}.

Fourth, we show by induction that \emph{for every $\phi$ in $\Phi_\FXL$, if $\models \FSD$, then $\vdash_\FXL \FSD$}.
The arguments go as follows:
\begin{itemize}

%%%%
\item

Let $\phi$ be a formula in $\Phi_\FXL$. Assume $\models \phi$. We want to show that $\vdash_\FXL \phi$. We put an induction on the modal depth $n$ of $\phi$.

%%%%
\item

Assume $n = 0$. Then, $\phi$ is a formula of the classical propositional logic. As $\FXL$ extends the classical propositional logic, $\vdash_\FXL \phi$.

%%%%
\item

Assume $n > 0$.

By the normal form lemma, there is $\phi'$ in $\Phi_\FXL$ meeting the following conditions: (1) $\models \phi$ if and only if $\models \phi'$, (2) $\vdash_\FXL \phi$ if and only if $\vdash_\FXL \phi'$, (3) $\phi'$ has the same modal depth with $\phi$, and (4) $\phi'$ is in a normal form $\FSD_0 \land \dots \land \FSD_n$, where every $\FSD_i$ is a standard disjunction.
Pick a $\FSD_i$. Then, $\models \FSD_i$. It suffices to show $\vdash_\FXL \FSD_i$.

Assume the modal depth of $\FSD_i$ is less than $n$. Then, by the inductive hypothesis, $\vdash_\FXL \FSD_i$.

Assume the modal depth of $\FSD_i$ is $n$. By the downward validity lemma, the validity-reduction condition of $\FSD$ is met. By the inductive hypothesis, the derivability-reduction condition of $\FSD$ is met. By the upward derivability lemma, $\vdash_\FXL \FSD_i$.

\end{itemize}

\paragraph{Remarks}

\emph{
Given the soundness and completeness of $\FXL$, we can see that for every standard disjunction $\FSD$, (1) $\FSD$ is valid if and only if the validity-reduction condition of $\FSD$ is met, and (2) $\FSD$ is derivable if and only if the derivability-reduction condition of $\FSD$ is met. Therefore, the names ``validity-reduction condition'' and ``derivability-reduction condition'' are correctly given.
}

\emph{
By the normal form lemma and the validity-reduction condition of standard disjunctions, we can reduce the validity problem of formulas of $\FXL$ to the validity problem of formulas of the classical propositional logic. It follows that $\FXL$ is decidable.
}

%%%%%%%%%%%%%%%%%%%%%%%%
%%%%%%%%%%%%%%%%%%%%%%%%
\section{Regular abstract game forms and their realization}
\label{section:A model construction technique}

In this section, we define regular abstract game forms and give a result about them, which will be used to show the downward validity lemmas for the four fragments.

Let $\FXL$ be a fragment among the four fragments.

We call a conjunction of literals in the classical propositional logic an \Fdefs{elementary conjunction}.

%%%%%%%%%%%%%%%%%%%%%%%%
%%%%%%%%%%%%%%%%%%%%%%%%
\begin{definition}[Regular abstract game forms]

An \defstyle{abstract game form} is a tuple $\FAGF = (\FAC_0, \Fforce)$ meeting the following conditions, where for every $\FAA \subseteq \FAG$, $\FJA^0_\FAA$ is the set of joint actions of $\FAA$ with respect to $\FAC_0$, and $\FJA^0 = \FBC \{\FJA^0_ \FAA \mid \FAA \subseteq \FAG\}$:
\begin{enumerate}[label=(\arabic*),leftmargin=3.33em]
\item $\FAC_0$ is a nonempty set of actions;
\item $\Fforce$ is a function from $\FJA^0$ to $\powerset{\Phi_{\FXL}}$.
\end{enumerate}

An abstract game form $\FAGF = (\FAC_0, \Fforce)$ is \Fdefs{regular} if:
\begin{enumerate}[label=(\arabic*),leftmargin=3.33em]
\item for every $\sigma \in \FJA$, $\Fforce (\sigma)$ is satisfiable;
\item $\Fforce$ is \Fdefs{monotonic}: for all $\FAA, \FAA' \subseteq \FAG$ such that $\FAA \subseteq \FAA'$, for all $\ja{\FAA} \in \FJA^0_\FAA$ and $\ja{\FAA'} \in \FJA^0_{\FAA'}$ such that $\ja{\FAA} \subseteq \ja{\FAA'}$, $\Fforce (\ja{\FAA}) \subseteq \Fforce (\ja{\FAA'})$.
\end{enumerate}

\end{definition}

Intuitively, $\Fforce (\sigma)$ consists of some formulas which $\sigma$ can \emph{ensure}.

%%%%%%%%%%%%%%%%%%%%%%%%%%%
%%%%%%%%%%%%%%%%%%%%%%%%%%%
\begin{definition}[Realization of regular abstract game forms and satisfiable elementary conjunctions]
\label{definition:Realization}

Let $\FAGF = (\FAC_0, \Fforce)$ be a regular abstract game form and $\gamma$ be a satisfiable elementary conjunction.

We say a pointed model $(\MM,s_0)$, where $\MM = (\FST, \FAC, \Faja, \Fout, \Flab)$, \Fdefs{realizes} $\FAGF$ and $\gamma$, if the following conditions are met:
\begin{enumerate}[label=(\arabic*),leftmargin=3.33em]

%%%%
\item 

$\FAC_0 \subseteq \FAC$;

%%%%
\item 

$\Faja (s_0, \FAA) = \FJA^0_\FAA$ for every $\FAA \subseteq \FAG$;

%%%%
\item 

$\MM, s_0 \Vdash \gamma$;

%%%%
\item 

$\ja{\FAA} \leadto_{s_0} \Fforce (\ja{\FAA})$ for all $\FAA \subseteq \FAG$ and $\ja{\FAA} \in \FJA^0_\FAA$.

\end{enumerate}

\end{definition}

%%%%%%%%%%%%%%%%%%%%%%%%%%%
%%%%%%%%%%%%%%%%%%%%%%%%%%%
\begin{theorem}[Realizability of regular abstract game forms and satisfiable elementary conjunctions]
\label{theorem:Realizability of abstract game forms}

Every regular abstract game form and satisfiable elementary conjunction is realized by a pointed concurrent game model.

\end{theorem}

%%%%%%%%%%%%%%%%%%%%%%%%
%%%%%%%%%%%%%%%%%%%%%%%%
\begin{proof}

Let $\FAGF = (\FAC_0, \Fforce)$ be an abstract game form and $\gamma$ be a satisfiable elementary conjunction.

-----------------------------------

Let $\{(\MM_{\ja{\FAG}}, s_{\ja{\FAG}}) \mid \ja{\FAG} \in \FJA^0_\FAG \}$, where $\MM_{\ja{\FAG}} = (\FST_{\ja{\FAG}}, \FAC_{\ja{\FAG}}, \Faja_{\ja{\FAG}}, \Fout_{\ja{\FAG}}, \Flab_{\ja{\FAG}})$, be a class of pointed models meeting the following conditions:
\begin{enumerate}[label=(\arabic*),leftmargin=3.33em]

%%%%
\item 

for all $\ja{\FAG} \in \FJA^0_\FAG$, $\MM_{\ja{\FAG}}, s_{\ja{\FAG}} \Vdash \Fforce (\ja{\FAG})$;

%%%%
\item 

all $\FST_{\ja{\FAG}}$ are pairwise disjoint;

%%%%
\item 

all $\FAC_{\ja{\FAG}}$ and $\FAC_0$ are pairwise disjoint.

\end{enumerate}

Let $s_0$ be a state not in any $\FST_{\ja{\FAG}}$.

Define a model $\MM = (\FST, \FAC, \Faja, \Fout, \Flab)$ as follows:
\begin{itemize}

%%%%
\item $\FST =\{s_0\} \cup \FBC \{\FST_{\ja{\FAG}} \mid \ja{\FAG} \in \FJA^0_\FAG\}$;

%%%%
\item $\FAC =\FAC_{0} \cup \FBC \{\FAC_{\ja{\FAG}} \mid \ja{\FAG} \in \FJA^0_\FAG\}$;

%%%%
\item for every $s \in \FST$ and $\FAA \subseteq \FAG$:
\[
\Faja (s, \FAA) =
\begin{cases}
\Faja_{\sigma_\FAG} (s, \FAA) & \text{if } s \in \FST_{\sigma_\FAG} \text{ for some } \ja{\FAG} \in \FJA^0_\FAG \\
\FJA^0_\FAA & \text{if } s = s_0
\end{cases}
\]

%%%%
\item $\Fout: \FST \times \FJA \to \powerset{\FST}$ is such that for every $s \in \FST$, $\FAA \subseteq \FAG$, and $\ja{\FAA} \in \FJA_\FAA$:
\[
\Fout (s, \ja{\FAA}) =
\]
\[
\begin{cases}
\{s_{\ja{\FAA}}\} & \text{if } \FAA = \FAG, s = s_0, \text{and } \ja{\FAA} \in \Faja_0 (s,\FAA) \\
\emptyset & \text{if } \FAA = \FAG, s = s_0, \text{and } \ja{\FAA} \notin \Faja_0 (s,\FAA) \\
\Fout_{\ja{\FAG}} (s, \sigma_\FAA) & \text{if } \FAA = \FAG, s \in \FST_{\sigma_\FAG} \text{ for some } \ja{\FAG} \in \FJA^0_\FAG, \text{and } \ja{\FAA} \in \Faja_{\sigma_\FAG} (s,\FAA) \\
\emptyset & \text{if } \FAA = \FAG, s \in \FST_{\sigma_\FAG} \text{ for some } \ja{\FAG} \in \FJA^0_\FAG, \text{and } \ja{\FAA} \notin \Faja_{\sigma_\FAG} (s,\FAA) \\
\FBC \{\Fout (s, \sigma_\FAG) \mid \sigma_\FAA \subseteq \sigma_\FAG\} & \text{otherwise} \\
\end{cases}
\]

%%%%
\item $\Flab: \FST \to \powerset{\FAP}$ is such that for every $s \in \FST$:
\[
\Flab (s) =
\begin{cases}
\Flab_{\sigma_\FAG} (s) & \text{if } s \in \FST_{\sigma_\FAG} \text{ for some } \ja{\FAG} \in \FJA^0_\FAG \\
\{p \mid \text{$p$ is a conjunct of $\gamma$}\} & \text{if } s = s_0 \\
\end{cases}
\]

\end{itemize}

It is easy to check that $(\MM,s_0)$ is a pointed concurrent game model.
Figure \ref{figure:grafted pointed cgm} explains $(\MM,s_0)$.

%%%%%%%%%%%%%%%%%%%%%%%%%%%%%%%%%%%%%%%%%%
%%%%%%%%%%%%%%%%%%%%%%%%%%%%%%%%%%%%%%%%%%
\begin{figure}
\begin{center}

\begin{tikzpicture}
[
->=stealth,
scale=1,
every node/.style={transform shape},
]
\tikzstyle{every state}=[minimum size=12mm]

\node[state] (s-1) {$s_{\sigma_\FAG^1}$};
\node[state,position=0:{66mm} from s-1] (s-n) {$s_{\sigma_\FAG^n}$};
\node[position=0:{33mm} from s-1] (s-ghost) {\dots};

\node[position=225:{20mm} from s-1] (s-1-1) {};
\node[position=315:{20mm} from s-1] (s-1-n) {};
\node[position=270:{11mm} from s-1] (s-1-ghost) {\dots};

\node[position=225:{20mm} from s-n] (s-n-1) {};
\node[position=315:{20mm} from s-n] (s-n-n) {};
\node[position=270:{11mm} from s-n] (s-n-ghost) {\dots};

\node[above=6mm] (a-s-1) at (s-1) {$\{\dots\}$};
\node[above=6mm] (a-s-n) at (s-n) {$\{\dots\}$};

\path
(s-1) edge [right] node {} (s-1-1)
(s-1) edge [left] node {} (s-1-n)

(s-n) edge [right] node {} (s-n-1)
(s-n) edge [left] node {} (s-n-n)
;

\node[position=90:{8mm} from s-1] (pm1) {$(\MM_{\sigma^1_\FAG},s_{\sigma^1_\FAG})$};
\node[position=90:{8mm} from s-n] (pmn) {$(\MM_{\sigma^n_\FAG},s_{\sigma^n_\FAG})$};

\end{tikzpicture}

\medskip

-----------------------------------------------

\medskip

\begin{tikzpicture}
[
->=stealth,
scale=1,
every node/.style={transform shape},
]
\tikzstyle{every state}=[minimum size=12mm]

\node[state] (s) {$s_0$};

\node[state,position=210:{35mm} from s] (s-1) {$s_{\sigma_\FAG^1}$};
\node[state,position=330:{35mm} from s] (s-n) {$s_{\sigma_\FAG^n}$};
\node[position=270:{17mm} from s] (s-ghost) {\dots};

\node[position=225:{20mm} from s-1] (s-1-1) {};
\node[position=315:{20mm} from s-1] (s-1-n) {};
\node[position=270:{11mm} from s-1] (s-1-ghost) {\dots};

\node[position=225:{20mm} from s-n] (s-n-1) {};
\node[position=315:{20mm} from s-n] (s-n-n) {};
\node[position=270:{11mm} from s-n] (s-n-ghost) {\dots};

\node[above=6mm] (a-s) at (s) {$\{q,r\}$};

\node[above=6mm] (a-s-1) at (s-1) {$\{\dots\}$};
\node[above=6mm] (a-s-n) at (s-n) {$\{\dots\}$};

\path
(s) edge [right] node {$\sigma_\FAG^1$} (s-1)
(s) edge [left] node {$\sigma_\FAG^n$} (s-n)

(s-1) edge [right] node {} (s-1-1)
(s-1) edge [left] node {} (s-1-n)

(s-n) edge [right] node {} (s-n-1)
(s-n) edge [left] node {} (s-n-n)
;

\node[position=90:{8mm} from s] (pm) {$(\MM,s_0)$};

\end{tikzpicture}

\caption{This figure explains the pointed model $(\MM,s_0)$ constructed in the proof for Theorem \ref{theorem:Realizability of abstract game forms}.
We assume $\FJA^0_\FAG = \{ \sigma^1_\FAG, \dots, \sigma^n_\FAG \}$ and $\gamma = \neg p \land q \land r$.
Above the separation line, there are $n$ pointed models: $(\MM_{\sigma_\FAG^1},s_{\sigma_\FAG^1}), \dots, (\MM_{\sigma_\FAG^n},s_{\sigma_\FAG^n})$.
For every $i$ such that $1 \leq i \leq n$, $\MM_{\sigma_\FAG^i},s_{\sigma_\FAG^i} \Vdash \Fforce (\sigma_\FAG^i)$.
The pointed model $(\MM,s_0)$ is depicted below the separation line.
In the model $\MM$:
(1) available joint actions of a coalition $\FAA$ at the state $s_0$ is specified by $\FJA^0_\FAA$; available joint actions of $\FAA$ at other states are as in the corresponding models;
(2) for every $i$ such that $1 \leq i \leq n$, the outcome state of the available joint action $\sigma^i_\FAG$ of $\FAG$ at the state $s_0$ is $s_{\sigma^i_\FAG}$;
outcome states of available joint actions of $\FAG$ at other states are as in the corresponding models;
outcome states of available joint actions of other coalitions at a state are derived as usual;
(3) at the state $s_0$, only $q$ and $r$ are true, which are conjuncts of $\neg p \land q \land r$; truth values of atomic propositions at other states are as in the corresponding models. 
}

\label{figure:grafted pointed cgm}

\end{center}
\end{figure}

-----------------------------------

We claim $(\MM,s_0)$ realizes $\FAGF$ and $\gamma$.

It is easy to see the first three conditions in the definition of realization of abstract game forms hold. It remains to show the fourth condition holds: $\ja{\FAA} \leadto_{s_0} \Fforce (\ja{\FAA})$ for every coalition $\FAA$ and $\ja{\FAA} \in \FJA^0_\FAA$.

Let $\ja{\FAG} \in \FJA^0_\FAG$. It can be easily shown that the \emph{generated sub-model} of $\MM_{\ja{\FAG}}$ at $s_{\ja{\FAG}}$ is also the \emph{generated sub-model} of $\MM$ at $s_{\ja{\FAG}}$. Then, for every $\phi \in \Phi_\FXL$, $\MM_{\ja{\FAG}}, s_{\ja{\FAG}} \Vdash \phi$ if and only if $\MM, s_{\ja{\FAG}} \Vdash \phi$.

Let $\FAA \subseteq \FAG$ and $\ja{\FAA} \in \FJA^0_\FAA$. Then $\ja{\FAA} \in \Faja (s_0,\FAA)$. Let $s \in \Fout (s_0, \ja{\FAA})$. We want to show $\MM, s \Vdash \Fforce (\sigma_\FAA)$.
Note: $\Fout (s_0, \ja{\FAA}) = 
\FBC \{\Fout (s_0, \sigma_\FAG) \mid \sigma_\FAG \in \FJA_\FAG \text{ and } \sigma_\FAA \subseteq \sigma_\FAG\} = 
\FBC \{\Fout (s_0, \sigma_\FAG) \mid \sigma_\FAG \in \Faja (s_0, \FAG) \text{ and } \sigma_\FAA \subseteq \sigma_\FAG\} = 
\FBC \{\{s_{\ja{\FAG}}\} \mid \sigma_\FAG \in \Faja (s_0, \FAG) \text{ and } \sigma_\FAA \subseteq \sigma_\FAG\} = 
\{s_{\ja{\FAG}} \mid \sigma_\FAG \in \Faja (s_0, \FAG) \text{ and } \sigma_\FAA \subseteq \sigma_\FAG\} =
\{s_{\ja{\FAG}} \mid \sigma_\FAG \in \FJA^0_\FAG \text{ and } \sigma_\FAA \subseteq \sigma_\FAG\}$.
Let $s = s_{\sigma_\FAG}$ for some $\sigma_\FAG \in \FJA^0_\FAG$. Then, $\MM_{\sigma_\FAG}, s_{\sigma_\FAG} \Vdash \Fforce (\sigma_\FAG)$.
Then, $\MM, s_{\sigma_\FAG} \Vdash \Fforce (\sigma_\FAG)$. Note $\Fforce$ is monotonic. Then, $\Fforce (\sigma_\FAA) \subseteq \Fforce (\sigma_\FAG)$. Then, $\MM, s_{\sigma_\FAG} \Vdash \Fforce (\sigma_\FAA)$.

\end{proof}

%%%%%%%%%%%%%%%%%%%%%%%%
%%%%%%%%%%%%%%%%%%%%%%%%
\section{Completeness of the liability fragment $\FCLn$ of $\FCL$}
\label{section:Completeness of the liability fragment of FCL}

%%%%%%%%%%%%%%%%%%%%%%%%
%%%%%%%%%%%%%%%%%%%%%%%%
\begin{definition}[An axiomatic system for $\FCLn$]\label{definition:An axiomatic system for FCLn}
~

\noindent Axioms: propositional tautologies in $\Phi_{\FCLn}$

\noindent Inference rules:

\begin{center}
\begin{tabular}{lll}
$\Frule_1$: & $\dfrac{\phi_1, \dots, \phi_n}
{\psi}$ & \parbox[t]{27em}{where $(\phi_1 \land \dots \land \phi_n) \rightarrow \psi$ is a propositional tautology} \vspace{10pt} \\
$\Frule_2$: & $\dfrac{\phi}
{\Fclod{\FAA} \phi}$ & \vspace{10pt} \\
$\Frule_3$: & $\dfrac{\Fclod{\FAA \cup \FBB} (\phi \lor \psi) \lor \chi}
{\Fclod{\FAA} \phi \lor \Fclod{\FBB} \psi \lor \chi}$ & \parbox[t]{25em}{where $\FAA \cap \FBB = \emptyset$}
\end{tabular}
\end{center}
\end{definition}

Note that implications are not well-formed in $\Phi_{\FCLn}$.

%%%%%%%%%%%%%%%%%%%%%%%%
%%%%%%%%%%%%%%%%%%%%%%%%
\begin{theorem}[Soundness of $\FCLn$]

The axiomatic system for $\FCLn$ given in Definition \ref{definition:An axiomatic system for FCLn} is sound with respect to the set of valid formulas in $\Phi_{\FCLn}$.

\end{theorem}

%%%%%%%%%%%%%%%%%%%%%%%%
%%%%%%%%%%%%%%%%%%%%%%%%
\begin{proof}

It suffices to show that all the axioms are valid and all the rules preserve validity.
It is easy to see that all the axioms are valid, and the first two rules preserve validity. 
We show the third rule preserves validity. It suffices to show that
$\Fclod{\FAA \cup \FBB} (\phi \lor \psi)
\rightarrow
\big(
\Fclod{\FAA} \phi \lor \Fclod{\FBB} \psi
\big)
$ is valid.
It suffices to show that 
$\big(
\neg \Fclod{\FAA} \phi \land \neg \Fclod{\FBB} \psi
\big)
\rightarrow
\neg \Fclod{\FAA \cup \FBB} (\phi \lor \psi)$
is valid.
It suffices to show that
$\big(
\Fclo{\FAA} \neg \phi \land \Fclo{\FBB} \neg \psi
\big)
\rightarrow
\Fclo{\FAA \cup \FBB} (\neg \phi \land \neg \psi)$
is valid, which is easy to see.

\end{proof}

%%%%%%%%%%%%%%%%%%%%%%%%%
%%%%%%%%%%%%%%%%%%%%%%%%%
\paragraph{Remarks}

\emph{Note that the formulas occurring in this proof are not formulas of $\Phi_\FCLn$, but this does not affect whether the proof works. This remark also applies to the soundness proofs for the other three fragments given later.}

%%%%%%%%%%%%%%%%%%%%%%%%
%%%%%%%%%%%%%%%%%%%%%%%%
\subsection{Standard disjunction and a normal form lemma for $\FCLn$}

We call a disjunction of literals in the classical propositional logic an \Fdefs{elementary disjunction}.
For every natural number $n$, we call $\FNI = \{x \in \mathbb{Z} \mid -n \leq x \leq -1\}$ a \Fdefs{set of negative indices}. Note that $\FNI$ might be empty.

%%%%%%%%%%%%%%%%%%%%%%%%
%%%%%%%%%%%%%%%%%%%%%%%%
\begin{definition}[Standard disjunctions for $\FCLn$]

Let $\gamma$ be an elementary disjunction and $\FNI$ be a set of negative indices.

A formula $\FSD$ in the form of $\gamma \vee \FBV_{i \in \FNI} \Fclod{\FAA_i} \phi_i$ in $\Phi_\FCLn$ is called a \defstyle{standard disjunction} for $\FCLn$ with respect to $\gamma$ and $\FNI$.

\end{definition}

The following example illustrates this definition.

%%%%%%%%%%%%%%%%%%%%%%%%
%%%%%%%%%%%%%%%%%%%%%%%%
\begin{example}[]
\label{example:Standard disjunctions}

Let $\gamma = \bot$ and $\FNI = \{-1, -2, -3\}$.
Then, the following formula is a standard disjunction for $\FCLn$ with respect to $\gamma$ and $\FNI$, where
$\FAA_{-1} = \{a\}$, $\phi_{-1} = p$, $\FAA_{-2} = \{b\}$, $\phi_{-2} = \neg p$, $\FAA_{-3} = \FAG$, $\phi_{-3} = q$:
\[
\FSD = \bot \lor \Fclod{\FAA_{-1}} \phi_{-1} \lor \Fclod{\FAA_{-2}} \phi_{-2} \lor \Fclod{\FAA_{-3}} \phi_{-3},
\]
that is,
\[
\FSD = \bot \lor \Fclod{a} p \lor \Fclod{b} \neg p \lor \Fclod{\FAG} q.
\]

\end{example}

\paragraph{Remarks}

\emph{
Strictly speaking, a standard disjunction for $\FCLn$ is determined by an elementary disjunction $\gamma$, a set of negative indices $\FNI$, a function from $\FNI$ to the set of coalitions, and a function from $\FNI$ to $\Phi_\FCLn$. Here, we make the two functions implicit. We will also do this in the sequel.
}

%%%%%%%%%%%%%%%%%%%%%%%%%%%
%%%%%%%%%%%%%%%%%%%%%%%%%%%
\begin{lemma}[Normal form for $\FCL$]
\label{lemma:normal-form CLn}

For every $\phi \in \Phi_{\FCLn}$, there is $\phi' \in \Phi_{\FCLn}$ such that (1) $\models \phi$ if and only if $\models \phi'$, (2) $\vdash_\FCLn \phi$ if and only if $\vdash_{\FCLn} \phi'$, (3) $\phi$ and $\phi'$ have the same modal depth, and (4) $\phi'$ is in the form of $\FSD_0 \land \dots \land \FSD_k$, where every $\FSD_i$ is a standard disjunction for $\FCLn$.

\end{lemma}

This lemma can be shown by (1) $\FCLn$ is an extension of the classical propositional logic, and (2) the soundness of $\FCLn$.

The following notion will be used to state the validity-reduction and the derivability-reduction conditions of standard disjunctions for $\FCLn$.

%%%%%%%%%%%%%%%%%%%%%%%%
%%%%%%%%%%%%%%%%%%%%%%%%
\begin{definition}[Neat negative index sets]

Let $\gamma$ be an elementary disjunction, $\FNI$ be a set of negative indices, and $\FSD = \gamma \vee \FBV_{i \in \FNI} \Fclod{\FAA_i} \phi_i$ be a standard disjunction for $\FCLn$ with respect to $\gamma$ and $\FNI$.

A subset $\FNI'$ of $\FNI$ is \Fdefs{neat} if for all $i, i' \in \FNI'$, if $i \neq i'$, then $\FAA_i \cap \FAA_{i'} = \emptyset$.

\end{definition}

%%%%%%%%%%%%%%%%%%%%%%%%%%%
%%%%%%%%%%%%%%%%%%%%%%%%%%%
\subsection{Downward validity lemma for $\FCLn$}

%%%%%%%%%%%%%%%%%%%%%%%%
%%%%%%%%%%%%%%%%%%%%%%%%
\begin{definition}[Validity-reduction condition of standard disjunction for $\FCLn$]

Let $\gamma$ be an elementary disjunction, $\FNI$ be a set of negative indices, and $\FSD = \gamma \vee \FBV_{i \in \FNI} \Fclod{\FAA_i} \phi_i$ be a standard disjunction for $\FCLn$ with respect to $\gamma$ and $\FNI$.

The \Fdefs{validity-reduction condition} of $\FSD$ is defined as follows: one of the following conditions holds:
\begin{enumerate}[label=(\alph*),leftmargin=3.33em]
\item $\vDash \gamma$;
\item there is a neat subset $\FNI'$ of $\FNI$ such that $\vDash \FBV_{i \in \FNI'} \phi_i$.
\end{enumerate}

\end{definition}

The following example illustrates this definition.

%%%%%%%%%%%%%%%%%%%%%%%%
%%%%%%%%%%%%%%%%%%%%%%%%
\begin{example}[]
\label{example:Standard disjunctions}

Let $\gamma = \bot$ and $\FNI = \{-1, -2, -3\}$.
Consider the following standard disjunction for $\FCLn$ with respect to $\gamma$ and $\FNI$, where
$\FAA_{-1} = \{a\}$, $\phi_{-1} = p$, $\FAA_{-2} = \{b\}$, $\phi_{-2} = \neg p$, $\FAA_{-3} = \FAG$, $\phi_{-3} = q$:
\[
\FSD = \bot \lor \Fclod{\FAA_{-1}} \phi_{-1} \lor \Fclod{\FAA_{-2}} \phi_{-2} \lor \Fclod{\FAA_{-3}} \phi_{-3},
\]
that is,
\[
\FSD = \bot \lor \Fclod{a} p \lor \Fclod{b} \neg p \lor \Fclod{\FAG} q.
\]

It can be easily verified that $\FSD$ is valid.
In fact, the validity-reduction condition of $\FSD$ holds.
Let $\FNI' = \{-1,-2\}$. Note $\FBV_{i \in \FNI'} \phi_i = p \lor \neg p$. Clearly, $\vDash p \lor \neg p$.

\end{example}

%%%%%%%%%%%%%%%%%%%%%%%%
%%%%%%%%%%%%%%%%%%%%%%%%
\begin{lemma}[Downward validity for $\FCLn$]
\label{lemma:Downward validity for FCLn}

Let $\gamma$ be an elementary disjunction, $\FNI$ be a set of negative indices, and $\FSD = \gamma \vee \FBV_{i \in \FNI} \Fclod{\FAA_i} \phi_i$ be a standard disjunction for $\FCLn$ with respect to $\gamma$ and $\FNI$.

Assumme $\vDash \FSD$.

Then, the validity-reduction condition of $\FSD$ is met; that is, one of the following conditions holds:
\begin{enumerate}[label=(\alph*),leftmargin=3.33em]
\item $\vDash \gamma$;
\item there is a neat subset $\FNI'$ of $\FNI$ such that $\vDash \FBV_{i \in \FNI'} \phi_i$.
\end{enumerate}

\end{lemma}

%%%%%%%%%%%%%%%%%%%%%%%%
%%%%%%%%%%%%%%%%%%%%%%%%
\begin{proof}
~

It is easy to see that the result holds if $\FNI = \emptyset$. Assume $\FNI \neq \emptyset$.

Assume the validity-reduction condition of $\FSD$ is not met. Then, (a) $\neg \gamma$ is satisfiable, and (b) for all neat subset $\FNI'$ of $\FNI$, $\FBW_{i \in \FNI'} \neg \phi_i$ is satisfiable.
It suffices to show $\neg \gamma \wedge \FBW_{i \in \FNI} \Fclo{\FAA_i} \neg \phi_i$ is satisfiable.

\paragraph{Remarks}

\emph{Note that formulas such as $\neg \gamma$ and $\neg \phi_i$ might not be formulas of $\Phi_\FCLn$, but this does not affect whether the proof works. This remark also applies to the remaining part of the proof and the proofs for the downward validity lemmas for the other three fragments given later.
}

--------------------

Define an \Fdefs{abstract game form $\FAGF = (\FAC_0, \Fforce)$ for $\FSD$} as follows, where for every $\FAA \subseteq \FAG$, $\FJA^0_\FAA$ is the set of joint actions of $\FAA$ with respect to $\FAC_0$:
\begin{itemize}
\item $\FAC_0 = \FNI$;

\item For every $\FCC \subseteq \FAG$ and $\sigma_\FCC \in \FJA^0_\FCC$, define:
\[
\Fcoin (\ja{\FCC}) = \{i \in \FNI \mid \FAA_i \subseteq \FCC \text{ and for all } a \in \FAA_i, \sigma_\FCC (a) = i\}
\]
\[
\Fforce (\ja{\FCC}) = \{\neg \phi_i \mid i \in \Fcoin ({\ja{\FCC}})\}
\]

\end{itemize}

%%%%%%%%%%%%%%%%%%%%%%%%
%%%%%%%%%%%%%%%%%%%%%%%%
\begin{example}[]
\label{example:Standard disjunctions}

This example illustrates the abstract game form with respect to $\FSD$.

Assume $\FAG = \{a,b\}$.
Let $\gamma = p$ and $\FNI = \{-1, -2\}$.
Consider the following standard disjunction for $\FCLn$ with respect to $\gamma$ and $\FNI$, where
$\FAA_{-1} = \{a\}$, $\phi_{-1} = q$, $\FAA_{-2} = \{b\}$, $\phi_{-2} = r$:
\[
\FSD = p \lor \Fclod{\FAA_{-1}} \phi_{-1} \lor \Fclod{\FAA_{-2}} \phi_{-2},
\]
that is,
\[
\FSD = p \lor \Fclod{a} q \lor \Fclod{b} r.
\]

It is easy to check that all four subsets of $\FNI$ are neat.
Then, it is easy to verify that (a) $\neg \gamma$ is satisfiable, and (b) for all neat subset $\FNI'$ of $\FNI$, $\FBW_{i \in \FNI'} \neg \phi_i$ is satisfiable.

The abstract game form $\FAGF = (\FAC_0, \Fforce)$ with respect to $\FSD$ is as follows:
\begin{itemize}

\item $\FAC_0 = \FNI$.

\item 

The sets of joint actions of coalitions are as follows:

$\FJA_\emptyset = \{\epsilon\}$;

$\FJA_a = \FJA_b = \FNI$;

$\FJA_\FAG = \{(-1,-1), (-1,-2), (-2, -1), (-2, -2) \}$.

\item

The function $\Fcoin$ from $\FJA$, which equals to $\bigcup \{\FJA_\FAA \mid \FAA \subseteq \FAG\}$, to the power set of $\FNI$ is as follows, where we use coalitions as subscripts of $\Fcoin$ to indicate whose joint actions its arguments are: 

$\Fcoin_\emptyset (\epsilon) = \emptyset$

$\Fcoin_a (-1) = \{-1\}$;

$\Fcoin_a (-2) = \emptyset$;

$\Fcoin_b (-1) = \emptyset$;

$\Fcoin_b (-2) = \{-2\}$;

$\Fcoin_\FAG (-1,-1) = \{-1\}$;

$\Fcoin_\FAG (-1,-2) = \{-1,-2\}$;

$\Fcoin_\FAG (-2,-1) = \{-2,-1\}$;

$\Fcoin_\FAG (-2,-2) = \{-2\}$.

\item 

The function $\Fforce$ from $\FJA$ to the power set of $\Phi_\FCL$ is as follows, where we use coalitions as subscripts of $\Fforce$ to indicate whose joint actions its arguments are:

$\Fforce_\emptyset (\epsilon) = \emptyset$;

$\Fforce_a (-1) = \{ \neg q \}$;

$\Fforce_a (-2) = \emptyset$;

$\Fforce_b (-1) = \emptyset$;

$\Fforce_a (-2) = \{ \neg r \}$;

$\Fforce_\FAG (-1,-1) = \{\neg q\}$;

$\Fforce_\FAG (-1,-2) = \{\neg q, \neg r\}$;

$\Fforce_\FAG (-2,-1) = \{\neg r, \neg q\}$;

$\Fforce_\FAG (-2,-2) = \{\neg r\}$.

\end{itemize}

\end{example}

We claim that $\FAGF$ is regular.

First, we show for all coalition $\FCC$ and $\ja{\FCC} \in \FJA^0_\FCC$, $\Fforce (\ja{\FCC})$ is satisfiable.
Let $\FCC$ be a coalition and $\ja{\FCC} \in \FJA^0_\FCC$. Note for all neat subset $\FNI'$ of $\FNI$, $\FBW_{i \in \FNI'} \neg \phi_i$ is satisfiable. It suffices to show that $\Fcoin (\ja{\FCC})$ is neat.
Let $i, i' \in \Fcoin (\ja{\FCC})$ such that $i \neq i'$.
Then, $\FAA_i \subseteq \FCC$, $\sigma_\FCC (a) = i$ for all $a \in \FAA_i$, $\FAA_{i'} \subseteq \FCC$, and $\sigma_\FCC (a) = i'$ for all $a \in \FAA_{i'}$.
Assume $\FAA_i \cap \FAA_{i'} \neq \emptyset$. Then there is an agent $a$ such that $a \in \FAA_i$ and $a \in \FAA_{i'}$. Then $\sigma_\FCC (a) = i$ and $\sigma_\FCC (a) = i'$. Then $i = i'$. We have a contradiction. Then, $\FAA_i \cap \FAA_{i'} = \emptyset$. Then, $\Fcoin (\ja{\FCC})$ is neat.

Second, we show that $\Fforce$ is monotonic. Let $\FAA, \FAA' \subseteq \FAG$, $\ja{\FAA} \in \FJA^0_\FAA$ and $\ja{\FAA'} \in \FJA^0_{\FAA'}$ such that $\ja{\FAA} \subseteq \ja{\FAA'}$. It suffices to show $\Fcoin (\ja{\FAA}) \subseteq \Fcoin (\ja{\FAA'})$.
Let $i \in \Fcoin (\ja{\FAA})$. Then $\FAA_i \subseteq \FAA$ and for all $a \in \FAA_i$, $\sigma_\FAA (a) = i$. Then $\FAA_i \subseteq \FAA'$ and for all $a \in \FAA_i$, $\sigma_{\FAA'} (a) = i$. Then $i \in \Fcoin (\ja{\FAA'})$. Then, $\Fcoin (\ja{\FAA}) \subseteq \Fcoin (\ja{\FAA'})$.

--------------------

Let $\gamma'$ be an elementary conjunction equivalent to $\neg \gamma$.

By Theorem \ref{theorem:Realizability of abstract game forms}, there is a pointed concurrent game model $(\MM,s_0)$ realizing $\FAGF$ and $\gamma'$. Let $\MM = (\FST, \FAC, \Faja, \Fout, \Flab)$. We want to show $\MM, s_0 \Vdash \neg \gamma \wedge \FBW_{i \in \FNI} \Fclo{\FAA_i} \neg \phi_i$.

By Definition \ref{definition:Realization}, $\MM, s_0 \Vdash \neg \gamma$.

Let $i \in \FNI$. It suffices to show $\MM, s_0 \Vdash \Fclo{\FAA_i} \neg \phi_i$.
Let $\ja{\FAA_i} \in \FJA^0_{\FAA_i}$ such that $\sigma_{\FAA_i} (a) = i$ for all $a \in \FAA_i$. Then $\neg \phi_i \in \Fforce (\ja{\FAA_i})$.
By Definition \ref{definition:Realization}, $\ja{\FAA_i} \leadto_{s_0} \neg \phi_i$. Then $\MM, s_0 \Vdash \Fclo{\FAA_i} \neg \phi_i$.

\end{proof}

%%%%%%%%%%%%%%%%%%%%%%%%%%%
%%%%%%%%%%%%%%%%%%%%%%%%%%%
\subsection{Upward derivability lemma for $\FCLn$}

%%%%%%%%%%%%%%%%%%%%%%%%
%%%%%%%%%%%%%%%%%%%%%%%%
\begin{definition}[Derivability-reduction condition of standard disjunction for $\FCLn$]

Let $\gamma$ be an elementary disjunction, $\FNI$ be a set of negative indices, and $\FSD = \gamma \vee \FBV_{i \in \FNI} \Fclod{\FAA_i} \phi_i$ be a standard disjunction for $\FCLn$ with respect to $\gamma$ and $\FNI$.

The \Fdefs{derivability-reduction condition} of $\FSD$ is defined as follows: one of the following conditions holds:
\begin{enumerate}[label=(\alph*),leftmargin=3.33em]
\item $\vdash_\FCLn \gamma$;
\item there is a neat subset $\FNI'$ of $\FNI$ such that $\vdash_\FCLn \FBV_{i \in \FNI'} \phi_i$.
\end{enumerate}

\end{definition}

%%%%%%%%%%%%%%%%%%%%%%%%
%%%%%%%%%%%%%%%%%%%%%%%%
\begin{lemma}[Upward derivability for $\FCLn$]
\label{lemma:Upward derivability for FCLn}

Let $\gamma$ be an elementary disjunction, $\FNI$ be a set of negative indices, and $\FSD = \gamma \vee \FBV_{i \in \FNI} \Fclod{\FAA_i} \phi_i$ be a standard disjunction for $\FCLn$ with respect to $\gamma$ and $\FNI$.

Assume the derivability-reduction condition of $\FSD$ is met; that is, one of the following conditions holds:
\begin{enumerate}[label=(\alph*),leftmargin=3.33em]
\item $\vdash_\FCLn \gamma$;
\item there is a neat subset $\FNI'$ of $\FNI$ such that $\vdash_\FCLn \FBV_{i \in \FNI'} \phi_i$.
\end{enumerate}

Then, $\vdash_\FCLn \FSD$.

\end{lemma}

%%%%%%%%%%%%%%%%%%%%%%%%
%%%%%%%%%%%%%%%%%%%%%%%%
\begin{proof}
~

It is easy to see that the result holds if $\FNI = \emptyset$. Assume $\FNI \neq \emptyset$.

Suppose (a). By application of the rule $\Frule_1$:
\[
\vdash_\FCLn \gamma \vee \FBV_{i \in \FNI} \Fclod{\FAA_i} \phi_i
\]

Suppose (b).

\begin{enumerate}[label=(\arabic*),leftmargin=3.33em]

%%%%%%%%%%
\item

By application of the rule $\Frule_2$:
\[
\vdash_\FCLn \Fclod{\FBC_{i \in \FNI'} \FAA_i} \FBV_{i \in \FNI'} \phi_i
\]

%%%%%%%%%%
\item

By repeated applications of the rule $\Frule_3$:
\[
\vdash_\FCLn \FBV_{i \in \FNI'} \Fclod{\FAA_i} \phi_i
\]

%%%%%%%%%%
\item

By application of the rule $\Frule_1$:
\[
\vdash_\FCLn \gamma \vee \FBV_{i \in \FNI} \Fclod{\FAA_i} \phi_i
\]

\end{enumerate}

\end{proof}

%%%%%%%%%%%%%%%%%%%%%%%%%%%
%%%%%%%%%%%%%%%%%%%%%%%%%%%
\subsection{Completeness of $\FCLn$ by induction}

%%%%%%%%%%%%%%%%%%%%%%%%
%%%%%%%%%%%%%%%%%%%%%%%%
\begin{theorem}[Completeness of $\FCLn$]\label{theorem:Completeness of FCLn}

The axiomatic system for $\FCLn$ given in Definition \ref{definition:An axiomatic system for FCLn} is complete with respect to the set of valid formulas in $\Phi_{\FCLn}$.

\end{theorem}

%%%%%%%%%%%%%%%%%%%%%%%%
%%%%%%%%%%%%%%%%%%%%%%%%
\begin{proof}
~

Let $\phi$ be a formula in $\Phi_{\FCLn}$. Assume $\models \phi$. We want to show $\vdash_\FCLn \phi$. We put an induction on the modal depth $n$ of $\phi$.

Assume $n = 0$. Then, $\phi$ is a propositional tautology. Then $\vdash_\FCLn \phi$.

Assume $n > 0$.

By Lemma \ref{lemma:normal-form CLn}, the normal form lemma for $\FCLn$, there is $\phi' \in \Phi_{\FCLn}$ such that (1) $\models \phi$ if and only if $\models \phi'$, (2) $\vdash_\FCLn \phi$ if and only if $\vdash_{\FCLn} \phi'$, (3) $\phi$ and $\phi'$ have the same modal depth, and (4) $\phi'$ is in the form of $\FSD_0 \land \dots \land \FSD_k$, where every $\FSD_i$ is a standard disjunction for $\FCLn$.
Then $\models \phi'$.

Let $i \leq k$. Then $\models \FSD_i$. It suffices to show $\vdash_\FCLn \FSD_i$.

Assume the modal depth of $\FSD_i$ is less than $n$. By the inductive hypothesis, $\vdash_\FCLn \FSD_i$.

Assume the modal depth of $\FSD_i$ is $n$. Let $\FSD_i = \gamma \vee \FBV_{i \in \FNI} \Fclod{\FAA_i} \phi_i$.

By Lemma \ref{lemma:Downward validity for FCLn}, the downward validity lemma for $\FCLn$, the validity-reduction condition of $\FSD_i$ is met; that is, one of the following conditions holds:
\begin{enumerate}[label=(\alph*),leftmargin=3.33em]
\item $\vDash \gamma$;
\item there is a neat subset $\FNI'$ of $\FNI$ such that $\vDash \FBV_{i \in \FNI'} \phi_i$.
\end{enumerate}

By the inductive hypothesis, the derivability-reduction condition of $\FSD_i$ is met; that is, one of the following conditions holds:
\begin{enumerate}[label=(\alph*),leftmargin=3.33em]
\item $\vdash_\FCLn \gamma$;
\item $\vdash_\FCLn \FBV_{i \in \FNI'} \phi_i$.
\end{enumerate}

By Lemma \ref{lemma:Upward derivability for FCLn}, the upward derivability lemma for $\FCLn$, $\vdash_\FCLn \FSD_i$.

\end{proof}

%%%%%%%%%%%%%%%%%%%%%%%%
%%%%%%%%%%%%%%%%%%%%%%%%
\section{Completeness of the ability fragment $\FCLp$ of $\FCL$}
\label{section:Completeness of the ability fragment of FCL}

%%%%%%%%%%%%%%%%%%%%%%%%
%%%%%%%%%%%%%%%%%%%%%%%%
\begin{definition}[An axiomatic system for $\FCLp$]\label{definition:An axiomatic system for FCLp}
~

\noindent Axioms: propositional tautologies in $\Phi_{\FCLn}$

\noindent Inference rules:

\begin{center}
\begin{tabular}{lll}
$\Frule_1$: & $\dfrac{\;\phi_1, \dots, \phi_n\;}
{\;\psi\;}$ & \parbox[t]{22em}{where $(\phi_1 \land \dots \land \phi_n) \rightarrow \psi$ is a propositional tautology} \vspace{10pt} \\
$\Frule_2$: & $\dfrac{\;\phi \;}
{\Fclo{\FAA} \phi }$ & \vspace{12pt} \\
$\Frule_3$: & $\dfrac{\Fclo{\FAA} (\phi_1 \vee \phi_2) \lor \chi}
{\Fclo{\FAA} \phi_1 \vee \Fclo{\FAG} \phi_2 \lor \chi}$ &
\end{tabular}
\end{center}
\end{definition}

%%%%%%%%%%%%%%%%%%%%%%%%
%%%%%%%%%%%%%%%%%%%%%%%%
\begin{theorem}[Soundness of $\FCLp$]\label{theorem:Soundness of the ability fragment of CL}

The axiomatic system for $\FCLp$ given in Definition \ref{definition:An axiomatic system for FCLp} is sound with respect to the set of valid formulas in $\Phi_{\FCLp}$.

\end{theorem}

%%%%%%%%%%%%%%%%%%%%%%%%
%%%%%%%%%%%%%%%%%%%%%%%%
\begin{proof}

It suffices to show that all the axioms are valid and all the rules preserve validity.
It is easy to see that all the axioms are valid, and the first two rules preserve validity. 
We show the third rule preserves validity. It suffices to show $\Fclo{\FAA} (\phi_1 \vee \phi_2) \rightarrow \big(\Fclo{\FAA} \phi_1 \vee \Fclo{\FAG} \phi_2 \big)$ is valid.
Let $(\MM,w)$, where $\MM = (\FST, \FAC, \Faja, \Fout, \Flab)$, be a pointed model such that $\MM,w \Vdash \Fclo{\FAA} (\phi_1 \vee \phi_2)$. Then, there is $\sigma_\FAA$ in $\Faja (w, \FAA)$ such that $\sigma_\FAA \leadto_w \phi_1 \lor \phi_2$.
Assume $\MM,w \not \Vdash \Fclo{\FAA} \phi_1$. It suffices to show $\MM,w \Vdash \Fclo{\FAG} \phi_2$.
Then $\sigma_\FAA \not \leadto_w \phi_1$.
Then there is $u \in \Fout (w, \sigma_\FAA)$ such that $\MM,u \Vdash \neg \phi_1$. Then $\MM,u \Vdash \phi_2$. Note there is $\sigma_\FAG \in \Faja (w, \FAG)$ such that $\sigma_\FAA \subseteq \sigma_\FAG$ and $\Fout (w, \sigma_\FAG) = \{u\}$. Then $\sigma_\FAG \leadto_w \phi_2$. Then $\MM,w \Vdash \Fclo{\FAG} \phi_2$.

\end{proof}

%%%%%%%%%%%%%%%%%%%%%%%%
%%%%%%%%%%%%%%%%%%%%%%%%
\subsection{Standard disjunctions and a normal form lemma for $\FCLp$}

For every natural number $n$, we call $\FPI = \{x \in \mathbb{Z} \mid 1 \leq x \leq n\}$ a \Fdefs{set of positive indices}. Note that $\FPI$ might be empty.

%%%%%%%%%%%%%%%%%%%%%%%%
%%%%%%%%%%%%%%%%%%%%%%%%
\begin{definition}[Standard disjunctions for $\FCLp$]

Let $\gamma$ be an elementary disjunction and $\FPI$ be a set of positive indices.

A formula $\FSD$ in the form of $\gamma \vee \FBV_{j \in \FPI} \Fclo{\FBB_j} \psi_j$ in $\Phi_{\FCLp}$ is called a \Fdefs{standard disjunction} for $\FCLp$ with respect to $\gamma$ and $\FPI$.

\end{definition}

The following example illustrates this definition.

%%%%%%%%%%%%%%%%%%%%%%%%
%%%%%%%%%%%%%%%%%%%%%%%%
\begin{example}[]
\label{example:Standard disjunctions}

Let $\gamma = \bot$ and $\FPI = \{1,2\}$.
Then, the following formula is called a \Fdefs{standard disjunction} for $\FCLp$ with respect to $\gamma$ and $\FPI$, where
$\FBB_{1} = \FAG$, $\psi_1 = p$, $\FBB_{2} = \{a\}$, $\psi_2 = \neg p$:
\[
\FSD = \bot \lor \Fclo{\FBB_{1}} \psi_1 \lor \Fclo{\FBB_{2}} \psi_2,
\]
that is,
\[
\FSD = \bot \lor \Fclo{\FAG} p \lor \Fclo{a} \neg p.
\]

\end{example}

%%%%%%%%%%%%%%%%%%%%%%%%%%%
%%%%%%%%%%%%%%%%%%%%%%%%%%%
\begin{lemma}[Normal form for $\FCLp$]
\label{lemma:normal-form CLp}

For every $\phi \in \Phi_{\FCLp}$, there is $\phi' \in \Phi_{\FCLp}$ such that (1) $\models \phi$ if and only if $\models \phi'$, (2) $\vdash_\FCLp \phi$ if and only if $\vdash_{\FCLp} \phi'$, (3) $\phi$ and $\phi'$ have the same modal depth, and (4) $\phi'$ is in the form of $\FSD_0 \land \dots \land \FSD_k$, where every $\FSD_i$ is a standard disjunction for $\FCLp$.

\end{lemma}

This lemma can be shown by (1) $\FCLp$ is an extension of the classical propositional logic, and (2) the soundness of $\FCLp$.

The following notation will be used to state the validity-reduction and the derivability-reduction conditions of standard disjunctions for $\FCLp$: $\FPIz = \{j \in \FPI \mid \FBB_j = \FAG\}$.

%%%%%%%%%%%%%%%%%%%%%%%%%%%
%%%%%%%%%%%%%%%%%%%%%%%%%%%
\subsection{Downward validity lemma for $\FCLp$}

%%%%%%%%%%%%%%%%%%%%%%%%
%%%%%%%%%%%%%%%%%%%%%%%%
\begin{definition}[Validity-reduction condition of standard disjunctions for $\FCLp$]

Let $\gamma$ be an elementary disjunction, $\FPI$ be a set of positive indices, and $\FSD = \gamma \vee \FBV_{j \in \FPI} \Fclo{\FBB_j} \psi_j$ be a standard disjunction for $\FCLp$ with respect to $\gamma$ and $\FPI$.

The \Fdefs{validity-reduction condition} of $\FSD$ is defined as follows: one of the following conditions holds:
\begin{enumerate}[label=(\alph*),leftmargin=3.33em]
\item $\vDash \gamma$;
\item there is $j' \in \FPI$ such that $\vDash \FBV_{j \in \FPIz} \psi_{j} \lor \psi_{j'}$.
\end{enumerate}

\end{definition}

The following example illustrates this definition.

%%%%%%%%%%%%%%%%%%%%%%%%
%%%%%%%%%%%%%%%%%%%%%%%%
\begin{example}[]
\label{example:Standard disjunctions}

Let $\gamma = \bot$ and $\FPI = \{1,2\}$.
Consider the following standard disjunction for $\FCLp$ with respect to $\gamma$ and $\FPI$, where
$\FBB_{1} = \FAG$, $\psi_1 = p$, $\FBB_{2} = \{a\}$, $\psi_2 = \neg p$:
\[
\FSD = \bot \lor \Fclo{\FBB_{1}} \psi_1 \lor \Fclo{\FBB_{2}} \psi_2,
\]
that is,
\[
\FSD = \bot \lor \Fclo{\FAG} p \lor \Fclo{a} \neg p.
\]

It is easy to verify $\models \FSD$.
In fact, the validity-reduction condition of $\FSD$ is met. Note $\FPI^0 = \{1\}$. Let $j' = 2$. Note $\FBV_{j \in \FPIz} \psi_{j} \lor \psi_{j'} = p \lor \neg p$. Clearly, $\vDash p \lor \neg p$.

\end{example}

%%%%%%%%%%%%%%%%%%%%%%%%
%%%%%%%%%%%%%%%%%%%%%%%%
\begin{lemma}[Downward validity for $\FCLp$]\label{lemma:Downward validity for FCLp}

Let $\gamma$ be an elementary disjunction, $\FPI$ be a set of positive indices, and $\FSD = \gamma \vee \FBV_{j \in \FPI} \Fclo{\FBB_j} \psi_j$ be a standard disjunction for $\FCLp$ with respect to $\gamma$ and $\FPI$.

Assume $\vDash \FSD$.

Then, the validity-reduction condition of $\FSD$ is met; that is, one of the following conditions holds:
\begin{enumerate}[label=(\alph*),leftmargin=3.33em]
\item $\vDash \gamma$;
\item there is $j' \in \FPI$ such that $\vDash \FBV_{j \in \FPIz} \psi_{j} \lor \psi_{j'}$.
\end{enumerate}

\end{lemma}

%%%%%%%%%%%%%%%%%%%%%%%%
%%%%%%%%%%%%%%%%%%%%%%%%
\begin{proof}
~

It is easy to see that the result holds if $\FPI = \emptyset$. Assume $\FPI \neq \emptyset$.

Assume the validity-reduction of $\FSD$ is not met. Then, (a) $\neg \gamma$ is satisfiable, and (b) for every $j' \in \FPI$, $\FBW_{j \in \FPIz} \neg \psi_{j} \wedge \neg \psi_{j'}$ is satisfiable.
It suffices to show $\neg \gamma \land \FBW_{j \in \FPI} \Fclod{\FBB_j} \neg \psi_j$ is satisfiable.

------------------

Assume $\FPI = \{1, \dots, n\}$. Define an abstract game form $\FAGF = (\FAC_0, \Fforce)$ as follows, where for every $\FCC \subseteq \FAG$, $\FJA^0_\FCC$ is the set of joint actions of $\FCC$ with respect to $\FAC_0$:
\begin{itemize}
\item $\FAC_0 = \FPI$;
\item For every $\FCC \subseteq \FAG$ and $\ja{\FCC} \in \FJA^0_\FCC$, define:
\[
\FArej (\ja{\FCC}) = \big(\sum_{a \in \FCC} \sigma_\FCC (a) \bmod n \big) + 1
\]
\[
\force{\ja{\FCC}} =
\begin{cases}
\{\neg \psi_k\} \cup \{\neg \psi_j \mid j \in \FPIz\} &\text{if } \FCC = \FAG \\
\{\neg \psi_j \mid j \in \FPIz\} & \text{otherwise}
\end{cases}
\]
where $k = \FArej (\ja{\FCC})$
\end{itemize}

%%%%%%%%%%%%%%%%%%%%%%%%
%%%%%%%%%%%%%%%%%%%%%%%%
\begin{example}[]
\label{example:Standard disjunctions}

The following example illustrates the abstract game form with respect to $\FSD$.

Assume $\FAG = \{a,b\}$. Let $\gamma = \bot$ and $\FPI = \{1,2\}$.
Consider the following standard disjunction for $\FCLp$ with respect to $\gamma$ and $\FPI$, where
$\FBB_{1} = \FAG$, $\psi_1 = p$, $\FBB_{2} = \{a\}$, $\psi_2 = q$:
\[
\FSD = \bot \lor \Fclo{\FBB_{1}} \psi_1 \lor \Fclo{\FBB_{2}} \psi_2,
\]
that is,
\[
\FSD = \bot \lor \Fclo{\FAG} p \lor \Fclo{a} q.
\]

Note $\FPI^0 = \{1\}$. It is easy to check that (a) $\neg \gamma$ is satisfiable, and (b) for every $j' \in \FPI$, $\FBW_{j \in \FPIz} \neg \psi_{j} \wedge \neg \psi_{j'}$ is satisfiable.

The abstract game form $\FAGF = (\FAC_0, \Fforce)$ with respect to $\FSD$ is as follows:
\begin{itemize}
\item $\FAC_0 = \FPI$.

\item 

The sets of joint actions of coalitions are as follows:

$\FJA_\emptyset = \{\epsilon\}$;

$\FJA_a = \FJA_b = \FPI$;

$\FJA_\FAG = \{(1,1), (1,2), (2, 1), (2, 2)\}$.

\item 

The function $\FArej$ from $\FJA$, which equals to $\bigcup \{\FJA_\FAA \mid \FAA \subseteq \FAG\}$, to $\FNI$ is as follows, where we use coalitions as subscripts of $\FArej$ to indicate whose joint actions its arguments are: 

$\FArej_a (1) = (1 \bmod 2) + 1 = 2$;

$\FArej_a (2) = (2 \bmod 2) + 1 = 1$;

$\FArej_b (1) = (1 \bmod 2) + 1 = 2$;

$\FArej_b (2) = (2 \bmod 2) + 1 = 1$;

$\FArej_\FAG (1,1) = (2 \bmod 2) + 1 = 1$;

$\FArej_\FAG (1,2) = (3 \bmod 2) + 1 = 2$;

$\FArej_\FAG (2,1) = (3 \bmod 2) + 1 = 2$;

$\FArej_\FAG (2,2) = (4 \bmod 2) + 1 = 1$.

\item 

The function $\Fforce$ from $\FJA$ to $\Phi_\FCL$ is as follows, where we use coalitions as subscripts of $\Fforce$ to indicate whose joint actions its arguments are:

$\Fforce_a (1) = \{\neg p\}$;

$\Fforce_a (2) = \{\neg p\}$;

$\Fforce_b (1) = \{\neg p\}$;

$\Fforce_b (2) = \{\neg p\}$;

$\Fforce_\FAG (1,1) = \{\neg p\}$;

$\Fforce_\FAG (1,2) = \{\neg p, \neg q\}$;

$\Fforce_\FAG (2,1) = \{\neg p, \neg q\}$;

$\Fforce_\FAG (2,2) = \{\neg p\}$.

\end{itemize}

\end{example}

By (b), for every $\FCC \subseteq \FAG$ and $\ja{\FCC} \in \FJA^0_\FCC$, $\force{\ja{\FCC}}$ is satisfiable. It is easy to see that $\Fforce$ is monotonic. Then, $\FAGF$ is regular.

--------------------

Let $\gamma'$ be an elementary conjunction equivalent to $\neg \gamma$.

By Theorem \ref{theorem:Realizability of abstract game forms}, there is a pointed concurrent game model $(\MM,s_0)$ realizing $\FAGF$ and $\gamma'$. Let $\MM = (\FST, \FAC, \Faja, \Fout, \Flab)$. We want to show $\MM, s_0 \Vdash \neg \gamma \wedge \FBW_{j \in \FPI} \Fclod{\FAA_j} \neg \psi_j$.

By Definition \ref{definition:Realization}, $\MM, s_0 \Vdash \neg \gamma$.

Let $j \in \FPI$. It suffices to show $\MM, s_0 \Vdash \Fclod{\FAA_j} \neg \psi_j$.
Let $\ja{\FAA_{j}} \in \Faja (s_0, \FAA_j)$. Then, $\ja{\FAA_{j}} \in \FJA^0_{\FAA_j}$.
Consider the following two cases:
\begin{itemize}

%%%%
\item

\textbf{Case 1:} $\FAA_{j} = \FAG$. Then, $j\in \FPIz$ and $\neg \psi_j\in \force{\ja{\FAA_{j}}}$. By Definition \ref{definition:Realization}, $\ja{\FAA_{j}} \leadto_{s_0} \neg \psi_{j}$.

%%%%
\item

\textbf{Case 2:} $\FAA_{j} \neq \FAG$. Then, $\overline{\FAA_j} \neq \emptyset$. Then, there is $\ja{\overline{\FAA_{j}}} \in \FJA^0_{\overline{\FAA_{j}}}$ such that $\FArej (\ja{\FAA_{j}} \cup \ja{\overline{\FAA_{j}}}) = j$. Then, $\neg \psi_j \in \force{\ja{\FAA_{j}} \cup \ja{\overline{\FAA_{j}}}}$. By Definition \ref{definition:Realization}, $\ja{\FAA_{j}} \cup \ja{\overline{\FAA_{j}}} \leadto_{s_0} \neg \psi_{j}$.

\end{itemize}

Then, $\MM, s_0 \Vdash \Fclod{\FAA_j} \neg \psi_j$.

\end{proof}

%%%%%%%%%%%%%%%%%%%%%%%%%%%
%%%%%%%%%%%%%%%%%%%%%%%%%%%
\subsection{Upward derivability lemma for $\FCLp$}

%%%%%%%%%%%%%%%%%%%%%%%%
%%%%%%%%%%%%%%%%%%%%%%%%
\begin{definition}[Derivability-reduction condition of standard disjunction for $\FCLp$]

Let $\gamma$ be an elementary disjunction, $\FPI$ be a set of positive indices, and $\FSD = \gamma \vee \FBV_{j \in \FPI} \Fclo{\FBB_j} \psi_j$ be a standard disjunction for $\FCLp$ with respect to $\gamma$ and $\FPI$.

The \Fdefs{derivability-reduction condition} of $\FSD$ is defined as follows: one of the following conditions holds:
\begin{enumerate}[label=(\alph*),leftmargin=3.33em]
\item $\vdash_\FCLp \gamma$;
\item there is $j' \in \FPI$ such that $\vdash_\FCLp \FBV_{j \in \FPIz} \psi_{j} \lor \psi_{j'}$.
\end{enumerate}

\end{definition}

%%%%%%%%%%%%%%%%%%%%%%%%
%%%%%%%%%%%%%%%%%%%%%%%%
\begin{lemma}[Upward derivability for $\FCLp$]\label{lemma:Upward derivability for FCLp}

Let $\gamma$ be an elementary disjunction, $\FPI$ be a set of positive indices, and $\FSD = \gamma \vee \FBV_{j \in \FPI} \Fclo{\FBB_j} \psi_j$ be a standard disjunction for $\FCLp$ with respect to $\gamma$ and $\FPI$.

Assume the derivability-reduction condition of $\FSD$ is met; that is, one of the following conditions holds:
\begin{enumerate}[label=(\alph*),leftmargin=3.33em]
\item $\vdash_\FCLp \gamma$;
\item there is $j' \in \FPI$ such that $\vdash_\FCLp \FBV_{j \in \FPIz} \psi_{j} \lor \psi_{j'}$.
\end{enumerate}

Then $\vdash_\FCLp \gamma \lor \FBV_{j \in \FPI} \Fclo{\FBB_j} \psi_j$.

\end{lemma}

%%%%%%%%%%%%%%%%%%%%%%%%
%%%%%%%%%%%%%%%%%%%%%%%%
\begin{proof}
~

It is easy to see that the result holds if $\FPI = \emptyset$. Assume $\FPI \neq \emptyset$.

Suppose (a). By application of the rule $\Frule_1$:
\[
\vdash_\FCLp \gamma \lor \FBV_{j \in \FPI} \Fclo{\FBB_i} \psi_j
\]

Suppose (b).

\begin{enumerate}[label=(\arabic*),leftmargin=3.33em]

%%%%%%%%%%
\item

By application of the rule $\Frule_2$:
\[
\vdash_\FCLp \Fclo{\FBB_{j'}}(\psi_{j'} \vee \FBV_{j \in \FPIz} \psi_{j})
\]

%%%%%%%%%%
\item

By repeated applications of the rules $\Frule_3$:
\[
\vdash_\FCLp \Fclo{\FBB_{j'}} \psi_{j'} \vee \FBV_{j \in \FPIz} \Fclo{\FAG} \psi_{j}
\]

%%%%%%%%%%
\item

By application of the rule $\Frule_1$:
\[
\vdash_\FCLp \gamma \lor \FBV_{j \in \FPI} \Fclo{\FBB_j} \psi_j
\]

\end{enumerate}

\end{proof}

%%%%%%%%%%%%%%%%%%%%%%%%%%%
%%%%%%%%%%%%%%%%%%%%%%%%%%%
\subsection{Completeness of $\FCLp$ by induction}

%%%%%%%%%%%%%%%%%%%%%%%%
%%%%%%%%%%%%%%%%%%%%%%%%
\begin{theorem}[Completeness of $\FCLp$]
The axiomatic system for $\FCLp$ given in Definition \ref{definition:An axiomatic system for FCLp} is complete with respect to the set of valid formulas in $\Phi_{\FCLp}$.
\end{theorem}

%%%%%%%%%%%%%%%%%%%%%%%%
%%%%%%%%%%%%%%%%%%%%%%%%
\begin{proof}
~

Let $\phi$ be a formula in $\Phi_{\FCLp}$. Assume $\models \phi$. We want to show $\vdash_\FCLp \phi$. We put an induction on the modal depth $n$ of $\phi$.

Assume $n = 0$. Then, $\phi$ is a propositional tautology. Then, $\vdash_\FCLp \phi$.

Assume $n > 0$.

By Lemma \ref{lemma:normal-form CLp}, the normal form lemma for $\FCLp$, there is $\phi' \in \Phi_{\FCLp}$ such that (1) $\models \phi$ if and only if $\models \phi'$, (2) $\vdash_\FCLp \phi$ if and only if $\vdash_{\FCLp} \phi'$, (3) $\phi$ and $\phi'$ have the same modal depth, and (4) $\phi'$ is in the form of $\FSD_0 \land \dots \land \FSD_k$, where every $\FSD_i$ is a standard disjunction for $\FCLp$.

Then $\models \phi'$. Let $i \leq k$. Then, $\models \FSD_i$. It suffices to show $\vdash_\FCLp \FSD_i$.

Assume the modal depth of $\FSD_i$ is less than $n$. By the inductive hypothesis, $\vdash_\FCLp \FSD_i$.

Assume the modal depth of $\FSD_i$ is $n$. Let $\FSD_i = \gamma \vee \FBV_{i \in \FPI} \Fclo{\FAA_i} \phi_i$.

By Lemma \ref{lemma:Downward validity for FCLp}, the downward validity lemma for $\FCLp$, the validity-reduction condition of $\FSD_i$ is met; that is, one of the following conditions holds:
\begin{enumerate}[label=(\alph*),leftmargin=3.33em]
\item $\vDash \gamma$;
\item there is $j' \in \FPI$ such that $\vDash \FBV_{j \in \FPIz} \psi_{j} \lor \psi_{j'}$.
\end{enumerate}

By the inductive hypothesis, the derivability-reduction condition of $\FSD_i$ is met; that is, one of the following conditions holds:
\begin{enumerate}[label=(\alph*),leftmargin=3.33em]
\item $\vdash_\FCLp \gamma$;
\item $\vdash_\FCLp \FBV_{j \in \FPIz} \psi_{j} \lor \psi_{j'}$.
\end{enumerate}

By Lemma \ref{lemma:Upward derivability for FCLp}, the upward derivability lemma for $\FCLp$, $\vdash_\FCLp \FSD_i$.

\end{proof}

%%%%%%%%%%%%%%%%%%%%%%%%
%%%%%%%%%%%%%%%%%%%%%%%%
\section{Completeness of the liability fragment $\FCCSRn$ of $\FCCSR$}
\label{section:Completeness of the liability fragment of FCCSR}

%%%%%%%%%%%%%%%%%%%%%%%%
%%%%%%%%%%%%%%%%%%%%%%%%
\begin{definition}[An axiomatic system for $\FCCSRn$]\label{definition:An axiomatic system for FCCSRn}
~

\noindent Axioms: propositional tautologies in $\Phi_{\FCCSRn}$

\noindent Inference rules:

\begin{center}
\begin{tabular}{lll}
$\Frule_1$: & $\dfrac{\phi_1, \dots, \phi_n}
{\psi}$ & \parbox[t]{22em}{where $(\phi_1 \land \dots \land \phi_n) \rightarrow \psi$ is a propositional tautology} \vspace{10pt} \\
$\Frule_2$: & $\dfrac{\phi}
{\Fchanced{\FAA}{\phi}{\emptyset}{\bot}}$ & \vspace{10pt} \\
$\Frule_3$: & $\dfrac{\Fchanced{\FAA \cup \FBB}{\phi \lor \psi}{\emptyset}{\bot} \lor \chi}
{\Fchanced{\FAA}{\phi}{\emptyset}{\bot} \lor \Fchanced{\FBB}{\psi}{\emptyset}{\bot} \lor \chi}$ & \parbox[t]{22em}{where $\FAA \cap \FBB = \emptyset$} \vspace{10pt} \\
$\Frule_4$: & $\dfrac{\Fchanced{\FAA}{\phi}{\emptyset}{\bot} \lor \chi}
{\Fchanced{\FAA}{\phi}{\FBB}{\psi} \lor \chi}$ & \vspace{10pt} \\
$\Frule_5$: & $\dfrac{\Fchanced{\FAA \cup \FBB}{\phi \vee \psi}{\emptyset}{\bot} \lor \chi}
{\Fchanced{\FAA}{\phi}{\FBB}{\psi} \lor \chi}$ &
\end{tabular}
\end{center}

\end{definition}

%%%%%%%%%%%%%%%%%%%%%%%%
%%%%%%%%%%%%%%%%%%%%%%%%
\begin{theorem}[Soundness of $\FCCSRn$]

The axiomatic system for $\FCCSRn$ given in Definition \ref{definition:An axiomatic system for FCCSRn} is sound with respect to the set of valid formulas in $\Phi_{\FCCSRn}$.

\end{theorem}

%%%%%%%%%%%%%%%%%%%%%%%%
%%%%%%%%%%%%%%%%%%%%%%%%
\begin{proof}
~

It suffices to show that all the axioms are valid and all the rules preserve validity.
It is easy to see that all the axioms are valid, and the first, the second, and the fourth rules preserve validity.

We show that the third rule preserves validity.
It suffices to show
$
\Fchanced{\FAA \cup \FBB}{\phi \lor \psi}{\emptyset}{\bot}
\rightarrow
\big( \Fchanced{\FAA}{\phi}{\emptyset}{\bot} \lor \Fchanced{\FBB}{\psi}{\emptyset}{\bot} \big)
$
is valid, where $\FAA \cap \FBB = \emptyset$. It is easy to see that this formula is equivalent to $
\neg \big(
\Fchanced{\FAA}{\phi}{\emptyset}{\bot} \lor \Fchanced{\FBB}{\psi}{\emptyset}{\bot}
\big)
\rightarrow 
\neg \Fchanced{\FAA \cup \FBB}{\phi \lor \psi}{\emptyset}{\bot}
$, which is equivalent to $
\big(
\Fchance{\FAA}{\neg \phi}{\emptyset}{\top} \land \Fchance{\FBB}{\neg \psi}{\emptyset}{\top}
\big)
\rightarrow 
\Fchance{\FAA \cup \FBB}{\neg \phi \land \neg \psi}{\emptyset}{\top}
$, which is equivalent to 
$
\big(\Fclo{\FAA} \neg \phi \land \Fclo{\FBB} \neg \phi\big) \rightarrow \Fclo{\FAA \cup \FBB} (\neg \phi \land \neg \psi)
$, which can be easily shown to be valid.

We show the fifth rule preserves validity. It suffices to show $
\Fchanced{\FAA \cup \FBB}{\phi \vee \psi}{\emptyset}{\bot}
\rightarrow
\Fchanced{\FAA}{\phi}{\FBB}{\psi}
$ is valid.
It suffices to show $
\neg \Fchanced{\FAA}{\phi}{\FBB}{\psi}
\rightarrow
\neg \Fchanced{\FAA \cup \FBB}{\phi \vee \psi}{\emptyset}{\bot}
$ is valid.
It suffices to show
$
\Fchance{\FAA}{\neg \phi}{\FBB}{\neg \psi}
\rightarrow
\Fchance{\FAA \cup \FBB}{\neg \phi \wedge \neg \psi}{\emptyset}{\top}
$ is valid, which is easy to see.

\end{proof}

%%%%%%%%%%%%%%%%%%%%%%%%
%%%%%%%%%%%%%%%%%%%%%%%%
\subsection{Standard disjunctions and a normal form lemma for $\FCCSRn$}

%%%%%%%%%%%%%%%%%%%%%%%%
%%%%%%%%%%%%%%%%%%%%%%%%
\begin{definition}[Standard disjunctions for $\FCCSRn$]

Let $\gamma$ be an elementary disjunction and $\FNI$ be a set of negative indices.

A formula $\FSD$ in the form of $\gamma \vee \FBV_{i \in \FNI} \Fchanced{\FAA_i}{\phi_i}{\FBB_i}{\psi_i}$ in $\Phi_{\FCCSRn}$ is called a \Fdefs{standard disjunction} for $\FCCSRn$ with respect to $\gamma$ and $\FNI$.

\end{definition}

The following example illustrates this definition.

%%%%%%%%%%%%%%%%%%%%%%%%%%%
%%%%%%%%%%%%%%%%%%%%%%%%%%%
\begin{example}
\label{example: a valid CConStR_LI standard disjunction}

Let $\gamma = \bot $ and $\FNI =\{-1,-2,-3\}$. Then, the following formula is a standard disjunction for $\FCCSRn$ with respect to $\gamma$ and $\FNI$, where $\FAA_{-1}=\{a\}$, $\phi_{-1}=p$, $\FBB_{-1}= \emptyset$, $\psi_{-1} = \bot$,
$\FAA_{-2}=\{b\}$, $\phi_{-2}=q$, $\FBB_{-2}=\emptyset$, $\psi_{-2} = \bot$,
$\FAA_{-3}=\{c\}$, $\phi_{-3} = \bot$, $\FBB_{-3}=\{d\}$, $\psi_{-3}=\neg p \wedge \neg q$:
\[
\FSD = \bot \vee \Fchanced{\FAA_{-1}}{\phi_{-1}}{\FBB_{-1}}{\psi_{-1}} \vee \Fchanced{\FAA_{-2}}{\phi_{-2}}{\FBB_{-2}}{\psi_{-2}} \vee \Fchanced{\FAA_{-3}}{\phi_{-3}}{\FBB_{-3}}{\psi_{-3}},
\]
that is,
\[
\FSD = \bot \vee \Fchanced{a}{p}{\emptyset }{\bot}\vee \Fchanced{b}{q}{\emptyset }{\bot}\vee \Fchanced{c}{\bot}{d}{(\neg p \wedge \neg q)}.
\]

\end{example}

%%%%%%%%%%%%%%%%%%%%%%%%%%%
%%%%%%%%%%%%%%%%%%%%%%%%%%%
\begin{lemma}[Normal form for $\FCCSRn$]
\label{lemma:normal-form CCSRn}

For every $\phi \in \Phi_{\FCCSRn}$, there is $\phi' \in \Phi_{\FCCSRn}$ such that (1) $\models \phi$ if and only if $\models \phi'$, (2) $\vdash_\FCCSRn \phi$ if and only if $\vdash_{\FCCSRn} \phi'$, (3) $\phi$ and $\phi'$ have the same modal depth, and (4) $\phi'$ is in the form of $\FSD_0 \land \dots \land \FSD_k$, where every $\FSD_i$ is a standard disjunction for $\FCCSRn$.

\end{lemma}

This lemma can be shown by (1) $\FCCSRn$ is an extension of the classical propositional logic, and (2) the soundness of $\FCCSRn$.

The following notion will be used in stating the validity-reduction and the derivability-reduction conditions of standard disjunctions for $\FCCSRn$.

%%%%%%%%%%%%%%%%%%%%%%%%
%%%%%%%%%%%%%%%%%%%%%%%%
\begin{definition}[Neat tuples of negative index sets]

Let $\gamma$ be an elementary disjunction, $\FNI$ be a set of negative indices, and $\FSD = \gamma \vee \FBV_{i \in \FNI} \Fchanced{\FAA_i}{\phi_i}{\FBB_i}{\psi_i}$ be a standard disjunction for $\FCCSRn$ with respect to $\gamma$ and $\FNI$.

A tuple $(\FNI_1,\FNI_2)$, where $\FNI_1 \subseteq \FNI$ and $\FNI_2 \subseteq \FNI$, is \Fdefs{neat} if:
\begin{itemize}
\item for all $i, i' \in \FNI_1$, if $i \neq i'$, then $\FAA_i \cap \FAA_{i'} = \emptyset$;
\item for all $i, i' \in \FNI_2$, if $i \neq i'$, then $(\FAA_i \cup \FBB_i) \cap (\FAA_{i'} \cup \FBB_{i'}) = \emptyset$;
\item for all $i \in \FNI_1$ and $i' \in \FNI_2$, if $i \neq i'$, then $\FAA_{i} \cap (\FAA_{i'} \cup \FBB_{i'}) = \emptyset$.
\end{itemize}

\end{definition}

The following example illustrates this definition.

%%%%%%%%%%%%%%%%%%%%%%%%%%%
%%%%%%%%%%%%%%%%%%%%%%%%%%%
\begin{example}
\label{example: neatness of the valid CConStR_LI standard disjunction}

Consider the standard disjunction for $\FCCSRn$ given in example \ref{example: a valid CConStR_LI standard disjunction}. Let $\FNI_1=\{-1,-2\}$ and $\FNI_2=\{-3\}$. It is easy to check $(\FNI_1,\FNI_2)$ is neat.

\end{example}

%%%%%%%%%%%%%%%%%%%%%%%%%%%
%%%%%%%%%%%%%%%%%%%%%%%%%%%
\subsection{Downward validity lemma for $\FCCSRn$}

%%%%%%%%%%%%%%%%%%%%%%%%
%%%%%%%%%%%%%%%%%%%%%%%%
\begin{definition}[Validity-reduction condition of standard disjunctions for $\FCCSRn$]

Let $\gamma$ be an elementary disjunction, $\FNI$ be a set of negative indices, and $\FSD = \gamma \vee \FBV_{i \in \FNI} \Fchanced{\FAA_i}{\phi_i}{\FBB_i}{\psi_i}$ be a standard disjunction for $\FCCSRn$ with respect to $\gamma$ and $\FNI$.

The \Fdefs{validity-reduction condition} of $\FSD$ is defined as follows: one of the following conditions holds:
\begin{enumerate}[label=(\alph*),leftmargin=3.33em]
\item $\models \gamma$;
\item there is $\FNI_1 \subseteq \FNI$ and $\FNI_2 \subseteq \FNI$ such that $(\FNI_1,\FNI_2)$ is neat and $\models \FBV_{i \in \FNI_1} \phi_i \vee \FBV_{i \in \FNI_2} (\phi_i \lor \psi_i)$.
\end{enumerate}

\end{definition}

The following example illustrates this definition.

%%%%%%%%%%%%%%%%%%%%%%%%
%%%%%%%%%%%%%%%%%%%%%%%%
\begin{example}

Let $\gamma = \bot $ and $\FNI =\{-1,-2,-3\}$. Consider the following standard disjunction for $\FCCSRn$ with respect to $\gamma$ and $\FNI$, where $\FAA_{-1}=\{a\}$, $\phi_{-1}=p$, $\FBB_{-1}= \emptyset$, $\psi_{-1} = \bot$,
$\FAA_{-2}=\{b\}$, $\phi_{-2}=q$, $\FBB_{-2}=\emptyset$, $\psi_{-2} = \bot$,
$\FAA_{-3}=\{c\}$, $\phi_{-3} = \bot$, $\FBB_{-3}=\{d\}$, $\psi_{-3}=\neg p \wedge \neg q$:
\[
\FSD = \bot \vee \Fchanced{\FAA_{-1}}{\phi_{-1}}{\FBB_{-1}}{\psi_{-1}} \vee \Fchanced{\FAA_{-2}}{\phi_{-2}}{\FBB_{-2}}{\psi_{-2}} \vee \Fchanced{\FAA_{-3}}{\phi_{-3}}{\FBB_{-3}}{\psi_{-3}},
\]
that is,
\[
\FSD = \bot \vee \Fchanced{a}{p}{\emptyset }{\bot}\vee \Fchanced{b}{q}{\emptyset }{\bot}\vee \Fchanced{c}{\bot}{d}{(\neg p \wedge \neg q)}.
\]

It is easy to verify $\models \FSD$. In fact, the validity-reduction condition of $\FSD$ is met.
Let $\FNI_1=\{-1,-2\}$ and $\FNI_2=\{-3\}$. As mentioned in Example \ref{example: neatness of the valid CConStR_LI standard disjunction}, $(\FNI_1,\FNI_2)$ is neat. Note $\FBV_{i \in \FNI_1} \phi_i \vee \FBV_{i \in \FNI_2} (\phi_i \lor \psi_i) = p \vee q\vee (\neg p \wedge \neg q)$. Clearly, $\vDash p \vee q\vee (\neg p \wedge \neg q)$.

\end{example}

%%%%%%%%%%%%%%%%%%%%%%%%
%%%%%%%%%%%%%%%%%%%%%%%%
\begin{lemma}[Downward validity for $\FCCSRn$]\label{lemma:Downward validity for FCCSRn}

Let $\gamma$ be an elementary disjunction, $\FNI$ be a set of negative indices, and $\FSD = \gamma \vee \FBV_{i \in \FNI} \Fchanced{\FAA_i}{\phi_i}{\FBB_i}{\psi_i}$ be a standard disjunction for $\FCCSRn$ with respect to $\gamma$ and $\FNI$.

Assume $\models \FSD$.

Then, the validity-reduction condition of $\FSD$ is met; that is, one of the following conditions holds:
\begin{enumerate}[label=(\alph*),leftmargin=3.33em]
\item $\vDash \gamma$;
\item there is $\FNI_1 \subseteq \FNI$ and $\FNI_2 \subseteq \FNI$ such that $(\FNI_1,\FNI_2)$ is neat and $\vDash \FBV_{i \in \FNI_1} \phi_i \vee \FBV_{i \in \FNI_2} (\phi_i \lor \psi_i)$.
\end{enumerate}

\end{lemma}

%%%%%%%%%%%%%%%%%%%%%%%%
%%%%%%%%%%%%%%%%%%%%%%%%
\begin{proof}
~

It is easy to see that the result holds if $\FNI = \emptyset$. Assume $\FNI \neq \emptyset$.

Assume the validity-reduction condition of $\FSD$ is not met. Then, (a) $\neg \gamma$ is satisfiable, and (b) for all $\FNI_1 \subseteq \FNI$ and $\FNI_2 \subseteq \FNI$, if $(\FNI_1,\FNI_2)$ is neat, then $\FBW_{i \in \FNI_1} \neg \phi_i \wedge \FBW_{i \in \FNI_2} (\neg \phi_i \land \neg \psi_i)$ is satisfiable.
It suffices to show $\neg \gamma \wedge \FBW_{i \in \FNI} \Fchance{\FAA_i}{\neg \phi_i}{\FBB_i}{\neg \psi_i}$ is satisfiable.

--------------------

Define an abstract game form $\FAGF = (\FAC_0, \Fforce)$ as follows, where for every $\FAA \subseteq \FAG$, $\FJA^0_\FAA$ is the set of joint actions of $\FAA$ with respect to $\FAC_0$:
\begin{itemize}
\item $\FAC_0 = \FNI$;

\item For every $\FCC \subseteq \FAG$ and $\ja{\FCC} \in \FJA^0_\FCC$, define:
\[
\Fcoinf (\ja{\FCC}) = \{i \in \FNI \mid \FAA_i \subseteq \FCC \text{ and for all } a \in \FAA_i, \sigma_\FCC (a) = i\}
\]
\[
\Fcoinfs (\ja{\FCC}) = \{i \in \FNI \mid \FAA_i \cup \FBB_i \subseteq \FCC \text{ and for all } a \in \FAA_i \cup \FBB_i, \sigma_\FCC (a) = i\}
\]
\[
\force{\ja{\FCC}} = \{ \neg \phi_i \mid i \in \Fcoinf (\ja{\FCC}) \} \cup \{ \neg \phi_i \land \neg \psi_i \mid i \in \Fcoinfs (\ja{\FCC}) \}
\]

\end{itemize}

We claim $\FAGF$ is regular.

First, we show that for every coalition $\FCC$ and $\ja{\FCC} \in \FJA^0_\FCC$, $\Fforce (\ja{\FCC})$ is satisfiable. By (b), it suffices to show that for every coalition $\FCC$ and $\ja{\FCC} \in \FJA^0_\FCC$, $(\Fcoinf (\ja{\FCC}), \Fcoinfs (\ja{\FCC}))$ is neat. Let $\FCC$ be a coalition and $\ja{\FCC} \in \FJA^0_\FCC$.
\begin{itemize}
\item Let $i, i'$ be two different indices in $\Fcoinf (\ja{\FCC})$. Assume $\FAA_i \cap \FAA_{i'} \neq \emptyset$. Let $a \in \FAA_i \cap \FAA_{i'}$. Then $\sigma_\FCC (a) = i$ and $\sigma_\FCC (a) = i'$. There is a contradiction. Then $\FAA_i \cap \FAA_{i'} = \emptyset$.
\item Let $i, i'$ be two different indices in $\Fcoinfs (\ja{\FCC})$. Assume $(\FAA_i \cup \FBB_i) \cap (\FAA_{i'} \cup \FBB_{i'}) \neq \emptyset$. Let $a \in (\FAA_i \cup \FBB_i) \cap (\FAA_{i'} \cup \FBB_{i'})$. Then $\sigma_\FCC (a) = i$ and $\sigma_\FCC (a) = i'$. There is a contradiction. Then $(\FAA_i \cup \FBB_i) \cap (\FAA_{i'} \cup \FBB_{i'}) = \emptyset$.
\item Let $i \in \Fcoinf (\ja{\FCC})$ and $i' \in \Fcoinfs (\ja{\FCC})$ such that $i \neq i'$. Assume $\FAA_i \cap (\FAA_{i'} \cup \FBB_{i'}) \neq \emptyset$. Let $a \in \FAA_i \cap (\FAA_{i'} \cup \FBB_{i'})$. Then $\sigma_\FCC (a) = i$ and $\sigma_\FCC (a) = i'$. There is a contradiction. Then $\FAA_i \cap (\FAA_{i'} \cup \FBB_{i'}) = \emptyset$.
\end{itemize}

Second, we show that $\Fforce$ is monotonic. Let $\FCC$ and $\FCC'$ be two coalitions, $\ja{\FCC} \in \FJA^0_\FCC$ and $\ja{\FCC'} \in \FJA^0_{\FCC'}$ such that $\ja{\FCC}\subseteq \ja{\FCC'}$. It is easy to see $\Fcoinf (\ja{\FCC}) \subseteq \Fcoinf (\ja{\FCC'})$ and $\Fcoinfs (\ja{\FCC}) \subseteq \Fcoinfs (\ja{\FCC'})$. Then, $\Fforce (\ja{\FAA}) \subseteq \Fforce (\ja{\FAA'})$.

--------------------

Let $\gamma'$ be an elementary conjunction equivalent to $\neg \gamma$.

By Theorem \ref{theorem:Realizability of abstract game forms}, there is a pointed concurrent game model $(\MM,s_0)$ realizing $\FAGF$ and $\gamma'$. Let $\MM = (\FST, \FAC, \Faja, \Fout, \Flab)$. We want to show $\MM, s_0 \Vdash \neg \gamma \wedge \FBW_{i \in \FNI} \Fchance{\FAA_i}{\neg \phi_i}{\FBB_i}{\neg \psi_i}$.

By Definition \ref{definition:Realization}, $\MM, s_0 \Vdash \neg \gamma$.

Let $i \in \FNI$. We want to show $\MM, s_0 \Vdash \Fchance{\FAA_i}{\neg \phi_i}{\FBB_i}{\neg \psi_i}$.
Let $\ja{\FAA_i \cup \FBB_i} \in \FJA^0_{\FAA_i \cup \FBB_i}$ be such that $\sigma_{\FAA_i \cup \FBB_i} (a) = i$ for all $a \in \FAA_i \cup \FBB_i$. Let $\sigma_{\FAA_i} = \ja{\FAA_i \cup \FBB_i} |_{\FAA_i}$.
Then $i \in \Fcoinf (\ja{\FAA_i})$ and $\neg \phi_i \in \Fforce (\sigma_{\FAA_i})$. By Definition \ref{definition:Realization}, $\ja{\FAA_i} \leadto_{s_0} \neg \phi_i$.
Then $i \in \Fcoinfs (\ja{\FAA_i \cup \FBB_i})$ and $\neg \phi_i \land \neg \psi_i \in \Fforce (\sigma_{\FAA_i \cup \FBB_i})$. By Definition \ref{definition:Realization}, $\ja{\FAA_i \cup \FBB_i} \leadto_{s_0} \neg \phi_i \land \neg \psi_i$.
Then $\MM, s_0 \Vdash \Fchance{\FAA_i}{\neg \phi_i}{\FBB_i}{\neg \psi_i}$.

\end{proof}

%%%%%%%%%%%%%%%%%%%%%%%%%%%
%%%%%%%%%%%%%%%%%%%%%%%%%%%
\subsection{Upward derivability lemma for $\FCCSRn$}

%%%%%%%%%%%%%%%%%%%%%%%%
%%%%%%%%%%%%%%%%%%%%%%%%
\begin{definition}[Derivability-reduction condition of standard disjunctions for $\FCCSRn$]

Let $\gamma$ be an elementary disjunction, $\FNI$ be a set of negative indices, and $\FSD = \gamma \vee \FBV_{i \in \FNI} \Fchanced{\FAA_i}{\phi_i}{\FBB_i}{\psi_i}$ be a standard disjunction for $\FCCSRn$ with respect to $\gamma$ and $\FNI$.

The \Fdefs{derivability-reduction condition} of $\FSD$ is defined as follows: one of the following conditions holds:
\begin{enumerate}[label=(\alph*),leftmargin=3.33em]
\item $\vdash_\FCCSRn \gamma$;
\item there is $\FNI_1 \subseteq \FNI$ and $\FNI_2 \subseteq \FNI$ such that $(\FNI_1,\FNI_2)$ is neat and $\vdash_\FCCSRn \FBV_{i \in \FNI_1} \phi_i \vee \FBV_{i \in \FNI_2} (\phi_i \lor \psi_i)$.
\end{enumerate}

\end{definition}

%%%%%%%%%%%%%%%%%%%%%%%%
%%%%%%%%%%%%%%%%%%%%%%%%
\begin{lemma}[Upward derivability for $\FCCSRn$]\label{lemma:Upward derivability for FCCSRn}

Let $\gamma$ be an elementary disjunction, $\FNI$ be a set of negative indices, and $\FSD = \gamma \vee \FBV_{i \in \FNI} \Fchanced{\FAA_i}{\phi_i}{\FBB_i}{\psi_i}$ be a standard disjunction for $\FCCSRn$ with respect to $\gamma$ and $\FNI$.

Assume the derivability-reduction condition of $\FSD$ is met; that is, one of the following conditions holds:
\begin{enumerate}[label=(\alph*),leftmargin=3.33em]
\item $\vdash_\FCCSRn \gamma$;
\item there is $\FNI_1 \subseteq \FNI$ and $\FNI_2 \subseteq \FNI$ such that $(\FNI_1,\FNI_2)$ is neat and $\vdash_\FCCSRn \FBV_{i \in \FNI_1} \phi_i \vee \FBV_{i \in \FNI_2} (\phi_i \lor \psi_i)$.
\end{enumerate}

Then, $\vdash_\FCCSRn \gamma \vee \FBV_{i \in \FNI} \Fchanced{\FAA_i}{\phi_i}{\FBB_i}{\psi_i}$.

\end{lemma}

%%%%%%%%%%%%%%%%%%%%%%%%
%%%%%%%%%%%%%%%%%%%%%%%%
\begin{proof}
~

It is easy to see that the result holds if $\FNI = \emptyset$. Assume $\FNI \neq \emptyset$.

Suppose (a). By application of the rule $\Frule_1$:
\[
\vdash_\FCCSRn \gamma \vee \FBV_{i \in \FNI} \Fchanced{\FAA_i}{\phi_i}{\FBB_i}{\psi_i}
\]

Suppose (b).

\begin{enumerate}[label=(\arabic*),leftmargin=3.33em]

%%%%
\item

By application of the rule $\Frule_2$:
\[
\vdash_\FCCSRn \Fchanced
{\FBC_{i \in \FNI_1} \FAA_i \cup \FBC_{i \in \FNI_2} (\FAA_i \cup \FBB_i)}
{\FBV_{i \in \FNI_1} \phi_i \vee \FBV_{i \in \FNI_2} (\phi_i \lor \psi_i)}
{\emptyset}
{\bot}
\]

%%%%
\item

By repeated applications of the rule $\Frule_3$:
\[
\vdash_\FCCSRn \FBV_{i \in \FNI_1} \Fchanced{\FAA_i}{\phi_i}{\emptyset}{\bot}
\lor
\FBV_{i \in \FNI_2} \Fchanced{\FAA_i \cup \FBB_i}{\phi_i \lor \psi_i}{\emptyset}{\bot}
\]

%%%%
\item

By repeated applications of the rule $\Frule_4$:
\[
\vdash_\FCCSRn \FBV_{i \in \FNI_1} \Fchanced{\FAA_i}{\phi_i}{\FBB_i}{\psi_i}
\lor
\FBV_{i \in \FNI_2} \Fchanced{\FAA_i \cup \FBB_i}{\phi_i \lor \psi_i}{\emptyset}{\bot}
\]

%%%%
\item

By repeated applications of the rule $\Frule_5$:
\[
\vdash_\FCCSRn \FBV_{i \in \FNI_1} \Fchanced{\FAA_i}{\phi_i}{\FBB_i}{\psi_i}
\lor
\FBV_{i \in \FNI_2} \Fchanced{\FAA_i }{\phi_i}{\FBB_i}{\psi_i}
\]

%%%%
\item

By application of the rule $\Frule_1$:
\[
\vdash_\FCCSRn \gamma \vee \FBV_{i \in \FNI} \Fchanced{\FAA_i}{\phi_i}{\FBB_i}{\psi_i}
\]

\end{enumerate}

\end{proof}

%%%%%%%%%%%%%%%%%%%%%%%%%%%
%%%%%%%%%%%%%%%%%%%%%%%%%%%
\subsection{Completeness of $\FCCSRn$ by induction}

%%%%%%%%%%%%%%%%%%%%%%%%
%%%%%%%%%%%%%%%%%%%%%%%%
\begin{theorem}[Completeness of $\FCCSRn$]
The axiomatic system for $\FCCSRn$ given in Definition \ref{definition:An axiomatic system for FCCSRn} is complete with respect to the set of valid formulas in $\Phi_{\FCCSRn}$.
\end{theorem}

%%%%%%%%%%%%%%%%%%%%%%%%
%%%%%%%%%%%%%%%%%%%%%%%%
\begin{proof}
~

Let $\phi$ be a formula in $\Phi_{\FCCSRn}$. Assume $\models \phi$. We want to show $\vdash_\FCCSRn \phi$. We put an induction on the modal depth $n$ of $\phi$.

Assume $n = 0$. Then, $\phi$ is a propositional tautology. Then, $\vdash_\FCCSRn \phi$.

Assume $n > 0$.

By Lemma \ref{lemma:normal-form CCSRn}, the normal form lemma for $\FCCSRn$, there is $\phi' \in \Phi_{\FCCSRn}$ such that (1) $\models \phi$ if and only if $\models \phi'$, (2) $\vdash_\FCCSRn \phi$ if and only if $\vdash_{\FCCSRn} \phi'$, (3) $\phi$ and $\phi'$ have the same modal depth, and (4) $\phi'$ is in the form of $\FSD_0 \land \dots \land \FSD_k$, where every $\FSD_i$ is a standard disjunction for $\FCCSRn$.

Then $\models \phi'$. Let $i \leq k$. Then $\models \FSD_i$. It suffices to show $\vdash_\FCCSRn \FSD_i$.

Assume the modal depth of $\FSD_i$ is less than $n$. By the inductive hypothesis, $\vdash_\FCCSRn \FSD_i$.

Assume the modal depth of $\FSD_i$ is $n$. Let $\FSD_i = \gamma \vee \FBV_{i \in \FNI} \Fchanced{\FAA_i}{\phi_i}{\FBB_i}{\psi_i}$.

By Lemma \ref{lemma:Downward validity for FCCSRn}, the downward validity lemma for $\FCCSRn$, the validity-reduction condition of $\FSD_i$ is met; that is, one of the following conditions holds:
\begin{enumerate}[label=(\alph*),leftmargin=3.33em]
\item $\vDash \gamma$;
\item there is $\FNI_1 \subseteq \FNI$ and $\FNI_2 \subseteq \FNI$ such that $(\FNI_1,\FNI_2)$ is neat and $\vDash \FBV_{i \in \FNI_1} \phi_i \vee \FBV_{i \in \FNI_2} (\phi_i \lor \psi_i)$.
\end{enumerate}

By the inductive hypothesis, the derivability-reduction condition of $\FSD_i$ is met; that is, one of the following conditions holds:
\begin{enumerate}[label=(\alph*),leftmargin=3.33em]
\item $\vdash_\FCCSRn \gamma$;
\item $\vdash_\FCCSRn \FBV_{i \in \FNI_1} \phi_i \vee \FBV_{i \in \FNI_2} (\phi_i \lor \psi_i)$.
\end{enumerate}

By Lemma \ref{lemma:Upward derivability for FCCSRn}, the upward derivability lemma for $\FCCSRn$, $\vdash_\FCCSRn \FSD_i$.

\end{proof}

%%%%%%%%%%%%%%%%%%%%%%%%
%%%%%%%%%%%%%%%%%%%%%%%%
\section{Completeness of the ability fragment $\FCCSRp$ of $\FCCSR$}
\label{section:Completeness of the ability fragment of FCCSR}

%%%%%%%%%%%%%%%%%%%%%%%%
%%%%%%%%%%%%%%%%%%%%%%%%
\begin{definition}[An axiomatic system for $\FCCSRp$]\label{definition:An axiomatic system for FCCSRp}
~

\noindent Axioms: propositional tautologies in $\Phi_{\FCCSRp}$

\noindent Inference rules:

\begin{center}
\begin{tabular}{lll}
$\Frule_1$: & $\dfrac{\phi_1, \dots, \phi_n}
{\psi}$ & \parbox[t]{17em}{where $(\phi_1 \land \dots \land \phi_n) \rightarrow \psi$ is a propositional tautology} \vspace{10pt} \\
$\Frule_2$: & $\dfrac{\phi}
{\Fchance{\FAA}{\phi}{\emptyset}{\top}}$ & \vspace{12pt} \\
$\Frule_3$: & $\dfrac
{
\Fchance
{\FAA}
{\phi_{1} \vee \phi_{2}}
{\emptyset}
{\top}
\vee \chi
}
{
\Fchance{\FAA}
{\phi_{1}}
{\emptyset}
{\top}
\vee
\Fchance
{\FAG}
{\phi_{2}}
{\emptyset}
{\top}
\vee \chi
}$ & \vspace{12pt} \\
$\Frule_4$: & $\dfrac{\Fchance{\FAA }{\phi \wedge \psi}{\emptyset}{\top}\vee \chi}
{\Fchance{\FAA}{\phi }{\FBB}{\psi} \vee \chi}$ & \vspace{12pt} \\
$\Frule_5$: & $\dfrac
{
(\Fchance{\FAA \cup \FBB}{\phi \wedge \psi}{\emptyset}{\top}\wedge \Fchance{\emptyset}{\phi}{\emptyset}{\top}) \vee \chi
}
{\Fchance{\FAA}{\phi }{\FBB}{\psi }\vee \chi}$ & 
\end{tabular}
\end{center}
\end{definition}

%%%%%%%%%%%%%%%%%%%%%%%%
%%%%%%%%%%%%%%%%%%%%%%%%
\begin{theorem}[Soundness of $\FCCSRp$]

The axiomatic system for $\FCCSRp$ given in Definition \ref{definition:An axiomatic system for FCCSRp} is sound with respect to the set of valid formulas in $\Phi_{\FCCSRp}$.

\end{theorem}

%%%%%%%%%%%%%%%%%%%%%%%%
%%%%%%%%%%%%%%%%%%%%%%%%
\begin{proof}
~

It suffices to show that all the axioms are valid and all the rules preserve validity.
It is easy to see that all the axioms are valid, and the first two rules preserve validity. 

We show the third rule preserves validity.
It suffices to show that $\Fchance
{\FAA}
{\phi_{1} \vee \phi_{2}}
{\emptyset}
{\top}
\rightarrow
\big( \Fchance{\FAA}
{\phi_{1}}
{\emptyset}
{\top}
\vee
\Fchance
{\FAG}
{\phi_{2}}
{\emptyset}
{\top} \big)$ is valid, which is equivalent to $\Fclo{\FAA} (\phi_1 \vee \phi_2) \rightarrow \big(\Fclo{\FAA} \phi_1 \vee \Fclo{\FAG} \phi_2 \big)$, which is shown to be valid in the proof for Theorem \ref{theorem:Soundness of the ability fragment of CL}.

We show the fourth rule preserves validity.
It suffices to show that $\Fchance{\FAA }{\phi \wedge \psi}{\emptyset}{\top} \rightarrow \Fchance{\FAA}{\phi }{\FBB}{\psi}$ is valid.
Let $(\MM,w)$, where $\MM = (\FST, \FAC, \Faja, \Fout, \Flab)$, be a pointed model.
Assume $\MM,w \Vdash \Fchance{\FAA }{\phi \wedge \psi}{\emptyset}{\top}$. Then there is $\sigma_\FAA \in \Faja (w,\FAA)$ such that $\sigma_\FAA \leadto_w \phi \land \psi$. Let $\sigma_\FBB \in \Faja (w,\FBB)$. Then $\sigma_\FAA \uplus \sigma_\FBB \leadto_w \phi \land \psi$. Then $\sigma_\FAA \uplus \sigma_\FBB \leadto_w \psi$. Then $\MM,w \Vdash \Fchance{\FAA}{\phi }{\FBB}{\psi}$.

We show the fifth rule preserves validity.
It suffices to show $\big(\Fchance{\FAA \cup \FBB}{\phi \wedge \psi}{\emptyset}{\top}\wedge \Fchance{\emptyset}{\phi}{\emptyset}{\top}\big) \rightarrow \Fchance{\FAA}{\phi }{\FBB}{\psi}$ is valid.
Let $(\MM,w)$, where $\MM = (\FST, \FAC, \Faja, \Fout, \Flab)$, be a pointed model.
Assume $\MM,w \Vdash \Fchance{\FAA \cup \FBB}{\phi \wedge \psi}{\emptyset}{\top}\wedge \Fchance{\emptyset}{\phi}{\emptyset}{\top}$. Then there is $\sigma_{\FAA \cup \FBB} \in \Faja (w,\FAA \cup \FBB)$ such that $\sigma_{\FAA \cup \FBB} \leadto_w \phi \land \psi$, and $\emptyset \leadto_w \phi$.
Let $\sigma_\FAA$ be the restriction of $\sigma_{\FAA \cup \FBB}$ to $\FAA$ and $\sigma_\FBB$ be the restriction of $\sigma_{\FAA \cup \FBB}$ to $\FBB$. 
Note $\sigma_{\FAA \cup \FBB} = \sigma_\FAA \uplus \sigma_\FBB$.
Then $\sigma_\FAA \leadto_w \phi$ and $\sigma_{\FAA \uplus \FBB} \leadto_w \psi$. Then $\MM,w \Vdash \Fchance{\FAA}{\phi }{\FBB}{\psi}$.

\end{proof}

Now, we start to show the completeness of $\FCCSRp$. Unlike in the last three sections, here, we will give two downward validity lemmas: the first downward validity lemma and the second downward validity lemma, and the former will be used to show the latter.

%%%%%%%%%%%%%%%%%%%%%%%%
%%%%%%%%%%%%%%%%%%%%%%%%
\subsection{Standard disjunctions and a normal form lemma for $\FCCSRp$}

%%%%%%%%%%%%%%%%%%%%%%%%
%%%%%%%%%%%%%%%%%%%%%%%%
\begin{definition}[Standard disjunctions for $\FCCSRp$]

Let $\gamma$ be an elementary disjunction and $\FPI$ be a set of positive indices.

A formula $\FSD$ in the form of $\gamma \vee \FBV_{j \in \FPI} \Fchp{j}$ in $\Phi_{\FCCSRp}$ is called a \Fdefs{standard disjunction} for $\FCCSRp$ with respect to $\gamma$ and $\FPI$.

\end{definition}

The following example illustrates this definition.

%%%%%%%%%%%%%%%%%%%%%%%%%%%
%%%%%%%%%%%%%%%%%%%%%%%%%%%
\begin{example}
\label{example: standard disjunction for CCSRp}

Let $\gamma = \bot $ and $\FPI =\{1,2,3\}$. Then, the following formula is a standard disjunction for $\FCCSRn$ with respect to $\gamma$ and $\FPI$, where 
$\FAA_1= \emptyset$, $\phi_1= \top$, $\FBB_1=\FAG$, $\psi_1=p_1$,
$\FAA_2 = \emptyset$, $\phi_2=\neg p_1$, $\FBB_2=\FAG$, $\psi_2 = p_2$,
$\FAA_3 = \emptyset$, $\phi_3=\neg p_2$, $\FBB_3=\emptyset$, $\psi_3=\top$:
\[\FSD = \bot \vee \Fchance{\FAA_1}{\phi_1}{\FBB_1}{\psi_1} \vee \Fchance{\FAA_2}{\phi_2}{\FBB_2}{\psi_2} \vee \Fchance{\FAA_3}{\phi_3}{\FBB_3}{\psi_3} ,\]
that is,
\[\FSD = \bot \vee \Fchance{\emptyset}{\top}{\FAG}{p_1} \vee \Fchance{\emptyset}{\neg p_1}{\FAG}{p_2}\vee \Fchance{\emptyset}{\neg p_2}{\emptyset }{\top} .\]

\end{example}

%%%%%%%%%%%%%%%%%%%%%%%%%%%
%%%%%%%%%%%%%%%%%%%%%%%%%%%
\begin{lemma}[Normal form for $\FCCSRp$]
\label{lemma:normal-form CCSRp}
For every $\phi \in \Phi_{\FCCSRp}$, there is $\phi' \in \Phi_{\FCCSRp}$ such that (1) $\models \phi$ if and only if $\models \phi'$, (2) $\vdash_\FCCSRp \phi$ if and only if $\vdash_{\FCCSRp} \phi'$, (3) $\phi$ and $\phi'$ have the same modal depth, and (4) $\phi'$ is in the form of $\FSD_0 \land \dots \land \FSD_k$, where every $\FSD_i$ is a standard disjunction for $\FCCSRp$.
\end{lemma}

This lemma can be shown by (1) $\FCCSRp$ is an extension of the classical propositional logic, and (2) the soundness of $\FCCSRp$.

The following notations will be used to state the validity-reduction and the derivability-reduction conditions of standard disjunctions for $\FCCSRp$:
$\FPIf = \{j \in \FPI \mid \FAA_j = \FAG\}$ and $\FPIfs = \{j \in \FPI \mid \FAA_j \cup \FBB_j = \FAG\}$. Here, ``$\mathtt{f}$'' and ``$\mathtt{s}$'' indicate ``first'' and ``second'', respectively.
Clearly, $\FPIf \subseteq \FPIfs$.

%%%%%%%%%%%%%%%%%%%%%%%%%%%
%%%%%%%%%%%%%%%%%%%%%%%%%%%
\subsection{The first downward validity lemma for $\FCCSRp$}

%%%%%%%%%%%%%%%%%%%%%%%%
%%%%%%%%%%%%%%%%%%%%%%%%
\begin{definition}[The first validity-reduction condition of standard disjunctions of $\FCCSRp$]

Let $\gamma$ be an elementary disjunction, $\FPI$ be a set of positive indices, and $\FSD = \gamma \vee \FBV_{j \in \FPI} \Fchp{j}$ be a standard disjunction for $\FCCSRp$ with respect to $\gamma$ and $\FPI$.

The \Fdefs{first validity-reduction condition} of $\FSD$ is defined as follows: one of the following conditions holds:
\begin{enumerate}[label=(\alph*),leftmargin=3.33em]

\item 

$\models \gamma$;

\item 

for all $\FPI^{*} \subseteq \FPI$ such that $\FPIf \subseteq \FPI^{*} \subseteq \FPIfs$, one of the following conditions holds:
\begin{enumerate}[label=(\arabic*)]
\item $\vDash \FBV_{j \in \FPI^*} (\phi_j \land \psi_j)$;
\item there is $j' \in \FPI - \FPI^{*}$ such that $\vDash \FBV_{j \in \FPI^*} (\phi_j \land \psi_j) \vee (\phi_{j'} \land \psi_{j'})$;
\item there is $j' \in \FPI - \FPI^{*}$ such that $\vDash \FBV_{j \in \FPI^*} (\phi_j \land \psi_j) \vee \phi_{j'}$ and $\FAA_{j'} \cup \FBB_{j'}= \FAG$.
\end{enumerate}

\end{enumerate}

\end{definition}

The following example illustrates this definition.

%%%%%%%%%%%%%%%%%%%%%%%%%%%
%%%%%%%%%%%%%%%%%%%%%%%%%%%
\begin{example}
\label{example: a valid standard disjunction CCSRp}

Let $\gamma = \bot $ and $\FPI =\{1,2,3\}$. Consider the following standard disjunction for $\FCCSRp$ with respect to $\gamma$ and $\FPI$, where 
$\FAA_1= \emptyset$, $\phi_1= \top$, $\FBB_1=\FAG$, $\psi_1=p_1$,
$\FAA_2 = \emptyset$, $\phi_2=\neg p_1$, $\FBB_2=\FAG$, $\psi_2 = p_2$,
$\FAA_3 = \emptyset$, $\phi_3=\neg p_2$, $\FBB_3=\emptyset$, $\psi_3=\top$:
\[\FSD = \bot \vee \Fchance{\FAA_1}{\phi_1}{\FBB_1}{\psi_1} \vee \Fchance{\FAA_2}{\phi_2}{\FBB_2}{\psi_2} \vee \Fchance{\FAA_3}{\phi_3}{\FBB_3}{\psi_3} ,\]
that is,
\[\FSD = \bot \vee \Fchance{\emptyset}{\top}{\FAG}{p_1} \vee \Fchance{\emptyset}{\neg p_1}{\FAG}{p_2}\vee \Fchance{\emptyset}{\neg p_2}{\emptyset }{\top} .\]

It can be easily shown that $\FSD$ is valid.
In fact, the first validity-reduction condition of $\FSD$ is met.
Note $\FPIf = \emptyset$ and $\FPIfs = \{1,2\}$. Suppose $\FPIf \subseteq \FPI^{*} \subseteq \FPIfs$.

Assume $\FPI^* = \emptyset$. It can be seen that the condition (b3) holds: 
$\vDash \FBV_{j \in \FPI^*} (\phi_j \land \psi_j) \vee \phi_{1}$ and $\FAA_{1} \cup \FBB_{1}= \FAG$.

Assume $\FPI^* = \{1\}$. It can be seen that the condition (b3) holds:
$\vDash \FBV_{j \in \FPI^*} (\phi_j \land \psi_j) \vee \phi_{2}$ and $\FAA_{2} \cup \FBB_{2}= \FAG$.

Assume $\FPI^* = \{2\}$. It can be seen that the condition (b3) holds:
$\vDash \FBV_{j \in \FPI^*} (\phi_j \land \psi_j) \vee \phi_{1}$ and $\FAA_{1} \cup \FBB_{1}= \FAG$.

Assume $\FPI^* = \{1,2\}$. It can be checked that the condition (b2) holds:
$\vDash \FBV_{j \in \FPI^*} (\phi_j \land \psi_j) \vee (\phi_{3}\land \psi_{3})$.

\end{example}

%%%%%%%%%%%%%%%%%%%%%%%%
%%%%%%%%%%%%%%%%%%%%%%%%
\begin{lemma}[The first downward validity for $\FCCSRp$]
\label{lemma:1 implies 2}

Let $\gamma$ be an elementary disjunction, $\FPI$ be a set of positive indices, and $\FSD = \gamma \vee \FBV_{j \in \FPI} \Fchp{j}$ be a standard disjunction for $\FCCSRp$ with respect to $\gamma$ and $\FPI$.

Assume $\vDash \FSD$.

Then, the first validity-reduction condition of $\FSD$ is met; that is, one of the following conditions holds:
\begin{enumerate}[label=(\alph*),leftmargin=3.33em]

\item 

$\models \gamma$;

\item 

for all $\FPI^{*} \subseteq \FPI$ such that $\FPIf \subseteq \FPI^{*} \subseteq \FPIfs$, one of the following conditions holds:
\begin{enumerate}[label=(\arabic*)]
\item $\vDash \FBV_{j \in \FPI^*} (\phi_j \land \psi_j)$;
\item there is $j' \in \FPI - \FPI^{*}$ such that $\vDash \FBV_{j \in \FPI^*} (\phi_j \land \psi_j) \vee (\phi_{j'} \land \psi_{j'})$;
\item there is $j' \in \FPI - \FPI^{*}$ such that $\vDash \FBV_{j \in \FPI^*} (\phi_j \land \psi_j) \vee \phi_{j'}$ and $\FAA_{j'} \cup \FBB_{j'}= \FAG$.
\end{enumerate}

\end{enumerate}

\end{lemma}

%%%%%%%%%%%%%%%%%%%%%%%%
%%%%%%%%%%%%%%%%%%%%%%%%
\begin{proof}
~

It is easy to see that the result holds if $\FPI = \emptyset$. Assume $\FPI \neq \emptyset$.

Assume the first validity-reduction condition of $\FSD$ is not met. Then, the following conditions hold:
\begin{enumerate}[label=(\alph*),leftmargin=3.33em]

\item 

$\neg \gamma$ is satisfiable;

\item

there is $\FPI^{*} \subseteq \FPI$ such that $\FPIf \subseteq \FPI^{*} \subseteq \FPIfs$ and the following conditions are met:

\begin{enumerate}
\item $\FBW_{j \in \FPI^*} \neg (\phi_j \land \psi_j)$ is satisfiable;
\item for all $j' \in \FPI - \FPI^{*}$, $\FBW_{j \in \FPI^*} \neg (\phi_j \land \psi_j) \wedge \neg (\phi_{j'} \wedge \psi_{j'})$ is satisfiable;
\item for all $j' \in \FPI - \FPI^{*}$, $\FBW_{j \in \FPI^*} \neg (\phi_j \land \psi_j) \wedge \neg \phi_{j'}$ is satisfiable, or $\nega{\FAA_{j'} \cup \FBB_{j'}} \neq \emptyset $.
\end{enumerate}

\end{enumerate}

It suffices to show that $\neg \gamma \land \FBW_{j \in \FPI} \neg \Fchp{j}$ is satisfiable.

--------------------

Assume $\FPI = \{1, \dots, n\}$.
Define an abstract game form $\FAGF = (\FAC_0, \Fforce)$ as follows, where for every $\FAA \subseteq \FAG$, $\FJA^0_\FAA$ is the set of joint actions of $\FAA$ with respect to $\FAC_0$:
\begin{itemize}

%%%%
\item $\FAC_0 = \FPI$;

\item For every $\FCC \subseteq \FAG$ and $\ja{\FCC} \in \FJA^0_\FCC$, define:
\[
\FArej (\ja{\FCC}) = (\sum_{a \in \FCC} \sigma_\FCC (a) \bmod n) + 1
\]
\[
\Fforce (\ja{\FCC})=
\begin{cases}
\parbox[t]{9.5em}{$\{\neg (\phi_j \land \psi_j) \mid j \in \FPI^*\} \cup \{\neg \phi_{k},\neg (\phi_{k} \wedge \psi_{k})\}$} & \parbox[t]{16em}{if $\FCC = \FAG$ and $\FBW_{j \in \FPI^*} \neg (\phi_j \land \psi_j) \wedge \neg \phi_{k}$ is satisfiable} \vspace{10pt} \\
\parbox[t]{9.5em}{$\{\neg (\phi_j \land \psi_j) \mid j \in \FPI^*\} \cup \{\neg (\phi_{k} \wedge \psi_{k})\}$} & \parbox[t]{16em}{if $\FCC = \FAG, \FBW_{j \in \FPI^*} \neg (\phi_j \land \psi_j) \wedge \neg \phi_{k}$ is not satisfiable, and $\FBW_{j \in \FPI^*} \neg (\phi_j \land \psi_j) \wedge \neg (\phi_{k} \wedge \psi_{k})$ is satisfiable} \vspace{10pt} \\
\parbox[t]{9.5em}{$\{\neg (\phi_j \land \psi_j) \mid j \in \FPI^*\}$} & \parbox[t]{16em}{otherwise}
\end{cases}
\]
\[
\text{where $k = \FArej(\ja{\FCC})$}
\]

\end{itemize}

It is easy to see that $\Fforce$ is monotonic and for every $\FCC \subseteq \FAG$ and $\ja{\FCC} \in \FJA^0_\FCC$, $\Fforce (\ja{\FCC})$ is satisfiable. Then, $\FAGF$ is regular.

Let $\gamma'$ be an elementary conjunction equivalent to $\neg \gamma$.

By Theorem \ref{theorem:Realizability of abstract game forms}, there is a pointed concurrent game model $(\MM,s_0)$ realizing $\FAGF$ and $\gamma'$. Let $\MM = (\FST, \FAC, \Faja, \Fout, \Flab)$.

--------------------

By Definition \ref{definition:Realization}, $\MM,s_0 \Vdash \neg \gamma$.

It suffices to show for all $j' \in \FPI$, $\MM, s_0 \Vdash \neg \Fchp{j'}$.
Let $j' \in \FPI$, $\ja{\FAA_{j'}} \in \FJA^0_{\FAA_{j'}}$ and $\ja{\FBB_{j'}} \in \FJA^0_{\FBB_{j'}}$. It suffices to show $\ja{\FAA_{j'}} \nleadto_{s_0} \phi_{j'}$ or $\ja{\FAA_{j'}} \uplus \ja{\FBB_{j'}} \nleadto_{s_0} \phi_{j'} \wedge \psi_{j'}$.

Consider the following two cases:
\begin{itemize}

%%%%
\item

\textbf{Case 1: $j' \in \FPI^{*}$.}

Then $\neg (\phi_{j'} \wedge \psi_{j'}) \in \Fforce (\ja{\FAA_{j'}} \uplus \ja{\FBB_{j'}})$. By Definition \ref{definition:Realization}, $\ja{\FAA_{j'}} \uplus \ja{\FBB_{j'}} \leadto_{s_0} \neg (\phi_{j'} \wedge \psi_{j'})$. Then $\ja{\FAA_{j'}} \uplus \ja{\FBB_{j'}} \nleadto_{s_0} \phi_{j'} \wedge \psi_{j'}$.

%%%%
\item

\textbf{Case 2: $j' \notin \FPI^{*}$.}

Note $\FPIf \subseteq \FPI^{*}$. Then, $\overline{\FAA_{j'}} \neq \emptyset$. By the assumption given at the beginning of this proof: (1) $\FBW_{j \in \FPI^*} \neg (\phi_{j} \land \psi_{j}) \wedge \neg (\phi_{j'} \wedge \psi_{j'})$ is satisfiable, and (2) $\FBW_{j \in \FPI^*} \neg (\phi_{j} \land \psi_{j}) \wedge \neg \phi_{j'}$ is satisfiable, or $\nega{\FAA_{j'} \cup \FBB_{j'}} \neq \emptyset$.

Consider the following two cases:
\begin{itemize}

%%%%
\item

\textbf{Case 2.1: $\FBW_{j \in \FPI^*} \neg (\phi_j \land \psi_j) \wedge \neg \phi_{j'}$ is satisfiable.}

By the definition of $\FArej$, we can get that $\nega{\FAA_{j'}}$ has a joint action $\ja{\nega{\FAA_{j'}}}$ in $\FJA^0_{\nega{\FAA_{j'}}}$ such that $\FArej (\ja{\FAA_{j'}} \cup \ja{\nega{\FAA_{j'}}}) = j'$. Then, $\neg \phi_{j'} \in \Fforce (\ja{\FAA_{j'}} \cup \ja{\nega{\FAA_{j'}}})$. By Definition \ref{definition:Realization}, $\ja{\FAA_{j'}} \cup \ja{\nega{\FAA_{j'}}} \leadto_{s_0} \neg \phi_{j'}$. Then, $\ja{\FAA_{j'}} \nleadto_{s_0} \phi_{j'}$.

%%%%
\item

\textbf{Case 2.2: $\nega{\FAA_{j'} \cup \FBB_{j'}} \neq \emptyset$.}

By the definition of $\FArej$, we can get that $\nega{\FAA_{j'} \cup \FBB_{j'}}$ has a joint action $\ja{\nega{\FAA_{j'} \cup \FBB_{j'}}}$ in $\FJA^0_{\nega{\FAA_{j'} \cup \FBB_{j'}}}$ such that $\FArej ((\ja{\FAA_{j'}} \uplus \ja{\FBB_{j'}}) \cup \ja{\nega{\FAA_{j'} \cup \FBB_{j'}}}) = j'$. Then, $\neg (\phi_{j'} \wedge \psi_{j'}) \in \Fforce ((\ja{\FAA_{j'}} \uplus \ja{\FBB_{j'}}) \cup \ja{\nega{\FAA_{j'} \cup \FBB_{j'}}})$. By Definition \ref{definition:Realization}, $\ja{\FAA_{j'}} \uplus \ja{\FBB_{j'}} \leadto_{s_0} \neg (\phi_{j'} \wedge \psi_{j'})$. Then, $\ja{\FAA_{j'}} \uplus \ja{\FBB_{j'}} \nleadto_{s_0} \phi_{j'} \wedge \psi_{j'}$.

\end{itemize}

\end{itemize}

\end{proof}

%%%%%%%%%%%%%%%%%%%%%%%%%%%
%%%%%%%%%%%%%%%%%%%%%%%%%%%
\subsection{The second downward validity lemma for $\FCCSRp$}

The second validity-reduction condition of standard disjunctions for $\FCCSRp$ involves a complex notion, which is defined as follows.

%%%%%%%%%%%%%%%%%%%%%%%%
%%%%%%%%%%%%%%%%%%%%%%%%
\begin{definition}[Semantically well-arranged sequences of sets of positive indices]

Let $\gamma$ be an elementary disjunction, $\FPI$ be a set of positive indices, and $\FSD = \gamma \vee \FBV_{j \in \FPI} \Fchp{j}$ be a standard disjunction for $\FCCSRp$ with respect to $\gamma$ and $\FPI$.

Let $\FPI^*_1, \FPI^*_2 \subseteq \FPI$. We say $\FPI^{*}_2$ is a \Fdefs{semantic expansion} of $\FPI^*_1$ if there is $j' \in \FPI - \FPI^{*}_1$ such that $\vDash \FBV_{j \in \FPI^*_1} (\phi_j \land \psi_j) \vee \phi_{j'}$ and $\FPI^*_2 = \FPI^*_1 \cup \{j'\}$.

Let $(\FPI^{*}_{k})_{0 \leq k \leq \FL}$ be a sequence of subsets of $\FPI$ for some $\FL \in \mathbb{N}$. We say $(\FPI^{*}_{k})_{0 \leq k \leq \FL}$ is \Fdefs{semantically well-arranged} if:
\begin{enumerate}[label=(\arabic*),leftmargin=3.33em]
\item $\FPI^*_0 = \FPIf$;
\item for every $k < \FL$, $\FPIf \subseteq \FPI^{*}_k \subseteq \FPIfs$;
\item for every $k < \FL$, $\FPI^*_{k+1}$ is a semantic expansion of $\FPI^*_{k}$;
\item $\vDash \FBV_{j \in \FPI^*_\FL} (\phi_j \land \psi_j)$.
\end{enumerate}

\end{definition}

An example for this definition will be given soon.

%%%%%%%%%%%%%%%%%%%%%%%%
%%%%%%%%%%%%%%%%%%%%%%%%
\begin{definition}[Second validity-reduction condition of standard disjunctions for $\FCCSRp$]

Let $\gamma$ be an elementary disjunction, $\FPI$ be a set of positive indices, and $\FSD = \gamma \vee \FBV_{j \in \FPI} \Fchp{j}$ be a standard disjunction for $\FCCSRp$ with respect to $\gamma$ and $\FPI$.

The \Fdefs{second validity-reduction condition} of $\FSD$ is defined as follows: one of the following conditions holds:
\begin{enumerate}[label=(\alph*),leftmargin=3.33em]
\item $\vDash \gamma$;
\item there is a semantically well-arranged sequence of subsets of $\FPI$.
\end{enumerate}

\end{definition}

The following example illustrates the two preceding definitions: 

%%%%%%%%%%%%%%%%%%%%%%%%
%%%%%%%%%%%%%%%%%%%%%%%%
\begin{example}

Let $\gamma = \bot $ and $\FPI =\{1,2,3\}$. Consider the following standard disjunction for $\FCCSRp$ with respect to $\gamma$ and $\FPI$, where 
$\FAA_1= \emptyset$, $\phi_1= \top$, $\FBB_1=\FAG$, $\psi_1=p_1$,
$\FAA_2 = \emptyset$, $\phi_2=\neg p_1$, $\FBB_2=\FAG$, $\psi_2 = p_2$,
$\FAA_3 = \emptyset$, $\phi_3=\neg p_2$, $\FBB_3=\emptyset$, $\psi_3=\top$:
\[\FSD = \bot \vee \Fchance{\FAA_1}{\phi_1}{\FBB_1}{\psi_1} \vee \Fchance{\FAA_2}{\phi_2}{\FBB_2}{\psi_2} \vee \Fchance{\FAA_3}{\phi_3}{\FBB_3}{\psi_3} ,\]
that is,
\[\FSD = \bot \vee \Fchance{\emptyset}{\top}{\FAG}{p_1} \vee \Fchance{\emptyset}{\neg p_1}{\FAG}{p_2}\vee \Fchance{\emptyset}{\neg p_2}{\emptyset }{\top} .\]

As mentioned in Example \ref{example: a valid standard disjunction CCSRp}, $\FSD$ is valid.
In fact, the second validity-reduction condition of $\FSD$ is met. It suffices to show $(\FPI^*_0, \FPI^*_1, \FPI^*_2, \FPI^*_3)$, where $\FPI^*_0 = \emptyset$, $\FPI^*_1 = \{1\}$, $\FPI^*_2 = \{1,2\}$, and $\FPI^*_3 = \{1,2,3\}$, is a semantically well-arranged sequence.
Note $\FPIf =\emptyset $ and $\FPIfs= \{1,2\}$.

It is easy to see that the conditions (1) and (2) in the definition of semantically well-arranged sequences hold.

Note $\FBV_{j \in \FPI^*_0} (\phi_j \land \psi_j) \vee \phi_1 = \bot \lor \top$, which is valid. Then, $\FPI^*_1$ is a semantic expansion of $\FPI^*_0$.

Note $\FBV_{j \in \FPI^*_1} (\phi_j \land \psi_j) \vee \phi_2 = (\top \wedge p_1) \vee \neg p_1$, which is valid. Then, $\FPI^*_2$ is a semantic expansion of $\FPI^*_1$.

Note $\FBV_{j \in \FPI^*_2} (\phi_j \land \psi_j) \vee \phi_3 = (\top \wedge p_1)\vee (\neg p_1\wedge p_2)\vee \neg p_2$, which is valid. Then, $\FPI^*_3$ is a semantic expansion of $\FPI^*_2$.

Then, the condition (3) holds.

Note $\FBV_{j \in \FPI^*_3} (\phi_j \land \psi_j) = (\top \wedge p_1)\vee (\neg p_1\wedge p_2)\vee (\neg p_2\wedge \top)$, which is valid. Then, the condition (4) holds.

\end{example}

%%%%%%%%%%%%%%%%%%%%%%%%
%%%%%%%%%%%%%%%%%%%%%%%%
\begin{lemma}[The second downward validity for $\FCCSRp$]
\label{lemma:Downward validity for FCCSRp}

Let $\gamma$ be an elementary disjunction, $\FPI$ be a set of positive indices, and $\FSD = \gamma \vee \FBV_{j \in \FPI} \Fchp{j}$ be a standard disjunction for $\FCCSRp$ with respect to $\gamma$ and $\FPI$.

Assume $\vDash \FSD$.

Then, the second validity-reduction condition of $\FSD$ is met; that is, one of the following conditions holds:
\begin{enumerate}[label=(\alph*),leftmargin=3.33em]
\item $\vDash \gamma$;
\item there is a semantically well-arranged sequence of subsets of $\FPI$.
\end{enumerate}

\end{lemma}

%%%%%%%%%%%%%%%%%%%%%%%%
%%%%%%%%%%%%%%%%%%%%%%%%
\begin{proof}
~

It is easy to see that the result holds if $\FPI = \emptyset$. Assume $\FPI \neq \emptyset$.

By Lemma \ref{lemma:1 implies 2}, the first downward validity lemma for $\FCCSRp$, the first validity-reduction of $\FSD$ is met; that is, one of the following conditions holds:
\begin{enumerate}[label=(\alph*),leftmargin=3.33em]

\item 

$\models \gamma$;

\item 

for all $\FPI^{*} \subseteq \FPI$ such that $\FPIf \subseteq \FPI^{*} \subseteq \FPIfs$, one of the following conditions holds:
\begin{enumerate}[label=(\arabic*)]
\item $\vDash \FBV_{j \in \FPI^*} (\phi_j \land \psi_j)$;
\item there is $j' \in \FPI - \FPI^{*}$ such that $\vDash \FBV_{j \in \FPI^*} (\phi_j \land \psi_j) \vee (\phi_{j'} \land \psi_{j'})$;
\item there is $j' \in \FPI - \FPI^{*}$ such that $\vDash \FBV_{j \in \FPI^*} (\phi_j \land \psi_j) \vee \phi_{j'}$ and $\FAA_{j'} \cup \FBB_{j'}= \FAG$.
\end{enumerate}

\end{enumerate}

It suffices to show if the condition (b) in the first validity-reduction of $\FSD$ holds, then there is a semantically well-arranged sequence of subsets of $\FPI$.
Assume the condition (b) in the first validity-reduction of $\FSD$ holds.

--------------------

Let $(\FPI^{*}_{k})_{0 \leq k \leq \FL}$ be a sequence of subsets of $\FPI$ for some $\FL \in \mathbb{N}$. We say $(\FPI^{*}_{k})_{0 \leq k \leq \FL}$ is \Fdefs{semantically pseudo-well-arranged} if:
\begin{itemize}
\item $\FPI^*_0 = \FPIf$;
\item for every $k \leq \FL$, $\FPIf \subseteq \FPI^{*}_k \subseteq \FPIfs$;
\item for every $k < \FL$, $\FPI^*_{k+1}$ is a semantic expansion of $\FPI^*_{k}$.
\end{itemize}

Let $\mathbf{SP}$ be the set of all semantically pseudo-well-arranged sequences of subsets of $\FPI$. Note that $(\FPIf)$ is a semantically pseudo-well-arranged sequence. Then, $\mathbf{SP}$ is not empty. Note that $\FPI$ is finite. Then, it is easy to see that $\mathbf{SP}$ is finite. Therefore, $\mathbf{SP}$ has a maximal sequence. Let $(\FPI^{*}_{k})_{0 \leq k \leq \FL}$ be such a maximal sequence.

--------------------

Assume $\vDash \FBV_{j \in \FPI^*_\FL} (\phi_j \land \psi_j)$. Then $(\FPI^{*}_{k})_{0 \leq k \leq \FL}$ is semantically well-arranged.

--------------------

Assume $\nvDash \FBV_{j \in \FPI^*_\FL} (\phi_j \land \psi_j)$. Since $\FPIf \subseteq \FPI^{*}_{\FL} \subseteq \FPIfs$, one of the following conditions is met:
\begin{enumerate}[label=(\alph*)]
\item $\vDash \FBV_{j \in \FPI^*_\FL} (\phi_j \land \psi_j)$;
\item there is $j' \in \FPI - \FPI^{*}_\FL$ such that $\vDash \FBV_{j \in \FPI^*} (\phi_j \land \psi_j) \vee (\phi_{j'} \land \psi_{j'})$;
\item there is $j' \in \FPI - \FPI^{*}_\FL$ such that $\vDash \FBV_{j \in \FPI^*_\FL} (\phi_j \land \psi_j) \vee \phi_{j'}$ and $\FAA_{j'} \cup \FBB_{j'}= \FAG$.
\end{enumerate}

Note (a) does not hold.

Assume (c). Look at the sequence $(\FPI^{*}_0, \dots, \FPI^*_\FL, \FPI^*_{\FL+1})$, where $\FPI^*_{\FL+1} = \FPI^*_\FL \cup \{j'\}$.
It is easy to see the following: 
\begin{itemize}
\item $\FPI^*_0 = \FPIf$;
\item for every $k \leq \FL+1$, $\FPIf \subseteq \FPI^{*}_k \subseteq \FPIfs$;
\item for every $k < \FL+1$, $\FPI_{k+1}$ is a semantic expansion of $\FPI_{k}$.
\end{itemize}

\noindent Then $(\FPI^{*}_0, \dots, \FPI^*_\FL, \FPI^*_{\FL+1})$ is semantically pseudo-well-arranged. Then $(\FPI^{*}_{k})_{0 \leq k \leq \FL}$ is not maximal. We have a contradiction.

Therefore, (b) holds. Note $\FBV_{j \in \FPI^*} (\phi_j \land \psi_j) \vee (\phi_{j'} \land \psi_{j'})$ implies $\FBV_{j \in \FPI^*_\FL} (\phi_j \land \psi_j) \vee \phi_{j'}$. Then $\vDash \FBV_{j \in \FPI^*_\FL} (\phi_j \land \psi_j) \vee \phi_{j'}$.
Look at the sequence $(\FPI^{*}_0, \dots, \FPI^*_\FL, \FPI^*_{\FL+1})$, where $\FPI^*_{\FL+1} = \FPI^*_\FL \cup \{j'\}$. It is easy to see the following:
\begin{itemize}
\item $\FPI^*_0 = \FPIf$;
\item for every $k < \FL+1$, $\FPIf \subseteq \FPI^{*}_k \subseteq \FPIfs$;
\item for every $k < \FL+1$, $\FPI^*_{k+1}$ is a semantic expansion of $\FPI^*_{k}$;
\item $\vDash \FBV_{j \in \FPI^*_{\FL+1}} (\phi_j \land \psi_j)$.
\end{itemize}

Then, $(\FPI^{*}_0, \dots, \FPI^*_\FL, \FPI^*_{\FL+1})$ is semantically well-arranged.

\end{proof}

%%%%%%%%%%%%%%%%%%%%%%%%%%%
%%%%%%%%%%%%%%%%%%%%%%%%%%%
\subsection{Upward derivability lemma for $\FCCSRp$}

%%%%%%%%%%%%%%%%%%%%%%%%
%%%%%%%%%%%%%%%%%%%%%%%%
\begin{definition}[Syntactically well-arranged sequences of sets of positive indices]

Let $\gamma$ be an elementary disjunction, $\FPI$ be a set of positive indices, and $\FSD = \gamma \vee \FBV_{j \in \FPI} \Fchp{j}$ be a standard disjunction for $\FCCSRp$ with respect to $\gamma$ and $\FPI$.

Let $\FPI^*_1, \FPI^*_2 \subseteq \FPI$.
We say $\FPI^{*}_2$ is a \Fdefs{syntactical expansion} of $\FPI^*_1$ if there is $j' \in \FPI - \FPI^{*}_1$ such that $\vdash_\FCCSRp \FBV_{j \in \FPI^*_1} (\phi_j \land \psi_j) \vee \phi_{j'}$ and $\FPI^*_2 = \FPI^*_1 \cup \{j'\}$.

Let $(\FPI^{*}_{k})_{0 \leq k \leq \FL}$ be a sequence of subsets of $\FPI$ for some $\FL \in \mathbb{N}$. We say $(\FPI^{*}_{k})_{0 \leq k \leq \FL}$ is \Fdefs{syntactically well-arranged} if:
\begin{enumerate}[label=(\arabic*),leftmargin=3.33em]
\item $\FPI^*_0 = \FPIf$;
\item for every $k < \FL$, $\FPIf \subseteq \FPI^{*}_k \subseteq \FPIfs$;
\item for every $k < \FL$, $\FPI^*_{k+1}$ is a syntactical expansion of $\FPI^*_{k}$;
\item $\vdash_\FCCSRp \FBV_{j \in \FPI^*_\FL} (\phi_j \land \psi_j)$.
\end{enumerate}

\end{definition}

%%%%%%%%%%%%%%%%%%%%%%%%
%%%%%%%%%%%%%%%%%%%%%%%%
\begin{definition}[Derivability-reduction condition of standard disjunctions for $\FCCSRp$]

Assume $\gamma$ is an elementary disjunction, $\FPI$ is a set of positive indices, and $\FSD = \gamma \vee \FBV_{j \in \FPI} \Fchp{j}$ is a standard disjunction for $\FCCSRp$ with respect to $\gamma$ and $\FPI$.

The \Fdefs{derivability-reduction condition} of $\FSD$ is defined as follows: one of the following conditions holds:
\begin{enumerate}[label=(\alph*),leftmargin=3.33em]
\item $\vdash_\FCCSRp \gamma$;
\item there is a syntactically well-arranged sequence of subsets of $\FPI$.
\end{enumerate}

\end{definition}

%%%%%%%%%%%%%%%%%%%%%%%%
%%%%%%%%%%%%%%%%%%%%%%%%
\begin{lemma}[Upward derivability for $\FCCSRp$]\label{lemma:Upward derivability for FCCSRp}

Let $\gamma$ be an elementary disjunction, $\FPI$ be a set of positive indices, and $\FSD = \gamma \vee \FBV_{j \in \FPI} \Fchp{j}$ be a standard disjunction for $\FCCSRp$ with respect to $\gamma$ and $\FPI$.

Assume the derivability-reduction condition of $\FSD$ is met; that is, one of the following conditions holds:
\begin{enumerate}[label=(\alph*),leftmargin=3.33em]
\item $\vdash_\FCCSRp \gamma$;
\item there is a syntactically well-arranged sequence of subsets of $\FPI$.
\end{enumerate}

Then, $\vdash_\FCCSRp \gamma \vee \FBV_{j \in \FPI} \Fchp{j}$.

\end{lemma}

%%%%%%%%%%%%%%%%%%%%%%%%
%%%%%%%%%%%%%%%%%%%%%%%%
\begin{proof}
~

It is easy to see that the result holds if $\FPI = \emptyset$. Assume $\FPI \neq \emptyset$.

Suppose (a). By the rule $\Frule_1$, $\vdash_\FCCSRp \gamma \vee \FBV_{j \in \FPI} \Fchance{\FAA_i}{\phi_i}{\FBB_i}{\psi_i}$.

Suppose (b). Let $(\FPI^{*}_{k})_{0 \leq k \leq \FL}$ be a syntactically well-arranged sequence for some $\FL \in \mathbb{N}$. Then:
\begin{itemize}
\item $\FPI^*_0 = \FPIf$;
\item for every $k < \FL$, $\FPIf \subseteq \FPI^{*}_{k} \subseteq \FPIfs$;
\item for every $k < \FL$, there is $j_{k+1} \in \FPI - \FPI^{*}_{k}$ such that $\vdash_\FCCSRp \FBV_{j^{*} \in \FPI^*_k} (\phi_{j^{*}} \land \psi_{j^{*}}) \vee \phi_{j_{k+1}}$ and $\FPI^{*}_{k+1} = \FPI^{*}_{k} \cup \{j_{k+1}\}$;
\item $\vdash_\FCCSRp \FBV_{j^{*} \in \FPI^*_\FL} (\phi_{j^{*}} \land \psi_{j^{*}})$.
\end{itemize}

\textbf{Assume $\FL=0$.}

\begin{enumerate}[label=(\arabic*),leftmargin=3.33em]

%%%%
\item

Since $(\FPI^{*}_{k})_{0 \leq k \leq \FL}$ is a syntactically well-arranged sequence:
\[
\vdash_\FCCSRp
\FBV_{j^{*}\in \FPI^{*}_{0}}(\phi_{j^{*}}\wedge \psi_{j^{*}})
\]

%%%%
\item

By application of the rule $\Frule_2$:
\[
\vdash_\FCCSRp
\Fchance
{\FAG}
{\FBV_{j^{*}\in \FPI^{*}_{0}}(\phi_{j^{*}}\wedge \psi_{j^{*}})}
{\emptyset}
{\top}
\]

%%%%
\item

By repeated applications of the rule $\Frule_3$:
\[
\vdash_\FCCSRp
\FBV_{j^{*}\in \FPI^{*}_{0}}
\Fchance
{\FAG }
{\phi_{j^{*}}\wedge \psi_{j^{*}}}
{\emptyset}
{\top}
\]

%%%%
\item

Note $\FAA_{j^{*}}=\FAG$ for all $j^{*}\in \FPI^{*}_{0}$. Then:
\[
\vdash_\FCCSRp
\FBV_{j^{*}\in \FPI^{*}_{0}}
\Fchance
{\FAA_{j^*}}
{\phi_{j^{*}}\wedge \psi_{j^{*}}}
{\emptyset}
{\top}
\]

%%%%
\item

By repeated applications of the rule $\Frule_4$:
\[
\vdash_\FCCSRp
\FBV_{j^{*}\in \FPI^{*}_{0}}
\Fchance
{\FAA_{j^*}}
{\phi_{j^{*}}}
{\FBB_{j^*}}
{\psi_{j^{*}}}
\]

%%%%
\item

By application of the rule $\Frule_1$:
\[
\vdash_\FCCSRp \gamma \vee \FBV_{j\in \FPI }\Fchp{j}
\]

\end{enumerate}

\textbf{Assume $\FL>0$.}

We claim that for every $k < \FL$:
\[
\vdash_\FCCSRp
\FBV_{j^{*}\in \FPI^{*}_{k}}
\Fchance
{\FAG}
{\phi_{j^{*}} \wedge \psi_{j^{*}}}
{\emptyset}
{\top}
\vee
\FBV_{j \in \FPI^{*}_{\FL} - \FPI^{*}_{k}} \Fchp{j}
\]

We put an induction on $k$ \emph{reversely}.

Assume $k = \FL - 1$.

\begin{enumerate}[label=(\arabic*),leftmargin=3.33em]

%%%%
\item

Since $(\FPI^{*}_{k})_{0 \leq k \leq \FL}$ is a syntactically well-arranged sequence:
\[
\vdash_\FCCSRp
\FBV_{j^{*}\in \FPI^{*}_{\FL}}(\phi_{j^{*}}\wedge \psi_{j^{*}})
\]

%%%%
\item

Note $\FPI^{*}_{\FL} = \FPI^{*}_{\FL-1} \cup \{j_{\FL}\}$. By application of the rule $\Frule_2$:
\[
\vdash_\FCCSRp
\Fchance
{\FAA_{j_{\FL}} }
{\FBV_{j^{*}\in \FPI^{*}_{\FL}}(\phi_{j^{*}}\wedge \psi_{j^{*}})}
{\emptyset}
{\top}
\]

%%%%
\item

Then:
\[
\vdash_\FCCSRp
\Fchance
{\FAA_{j_{\FL}}}
{\FBV_{j^{*} \in \FPI^{*}_{\FL-1}} (\phi_{j^{*}} \wedge \psi_{j^{*}}) \vee (\phi_{j_\FL} \wedge \psi_{j_\FL})}
{\emptyset}
{\top}
\]

%%%%
\item

By repeated applications of the rule $\Frule_3$:
\[
\vdash_\FCCSRp
\FBV_{j^{*}\in \FPI^{*}_{\FL-1}} \Fchance
{\FAG }
{\phi_{j^{*}} \wedge \psi_{j^{*}} }
{\emptyset}
{\top}
\vee
\Fchance
{\FAA_{j_{\FL}} }
{\phi_{j_{\FL}}\wedge \psi_{j_{\FL}}}
{\emptyset}
{\top}
\]

%%%%
\item

By application of the rule $\Frule_4$:
\[
\vdash_\FCCSRp
\FBV_{j^{*}\in \FPI^{*}_{\FL-1}} \Fchance
{\FAG }
{\phi_{j^{*}} \wedge \psi_{j^{*}} }
{\emptyset}
{\top}
\vee
\Fchance
{\FAA_{j_{\FL}}}
{\phi_{j_{\FL}}}
{\FBB_{j_{\FL}}}
{\psi_{j_{\FL}}}
\]

\end{enumerate}

Assume $k = m - 1$ for some $m$ such that $1 \leq m < \FL$.

\begin{enumerate}[label=(\arabic*),leftmargin=3.33em]

%%%%
\item

By the inductive hypothesis:
\[
\vdash_\FCCSRp
\FBV_{j^{*}\in \FPI^{*}_{m}} \Fchance
{\FAG }
{\phi_{j^{*}} \wedge \psi_{j^{*}} }
{\emptyset }
{\top }
\vee
\FBV_{j \in \FPI^{*}_{\FL}-\FPI^{*}_{m}}\Fchp{j}
\]

%%%%
\item

Note $\FPI^*_m = \FPI^*_{m-1} \cup \{j_m\}$. Then:
\[
\vdash_\FCCSRp
\FBV_{j^{*}\in \FPI^{*}_{m-1}}(\phi_{j^{*}}\wedge \psi_{j^{*}})
\vee
\phi_{j_{m}}
\]

%%%%
\item

By application of the rule $\Frule_2$:
\[
\vdash_\FCCSRp
\Fchance
{\emptyset }
{\FBV_{j^{*}\in \FPI^{*}_{m-1}}(\phi_{j^{*}} \wedge \psi_{j^{*}}) \vee \phi_{j_{m}}}
{\emptyset}
{\top}
\]

%%%%
\item

By repeated applications of the rule $\Frule_3$:
\[
\vdash_\FCCSRp
\FBV_{j^{*} \in \FPI^{*}_{m-1}}
\Fchance
{\FAG}
{\phi_{j^{*}} \wedge \psi_{j^{*}}}
{\emptyset}
{\top}
\vee
\Fchance
{\emptyset }
{\phi_{j_{m}} }
{\emptyset }
{\top }
\]

%%%%
\item

By application of the rule $\Frule_1$:
\[
\vdash_\FCCSRp
\FBV_{j^{*} \in \FPI^{*}_{m-1}}
\Fchance
{\FAG}
{\phi_{j^{*}} \wedge \psi_{j^{*}}}
{\emptyset}
{\top}
\vee
\Fchance
{\emptyset }
{\phi_{j_{m}} }
{\emptyset }
{\top }
\vee
\FBV_{j \in \FPI_{\FL}-\FPI^{*}_{m}}\Fchp{j}
\]

%%%%
\item

From (1):
\[
\vdash_\FCCSRp
\FBV_{j^{*} \in \FPI^{*}_{m-1}} \Fchance
{\FAG }
{\phi_{j^{*}} \wedge \psi_{j^{*}} }
{\emptyset }
{\top }
\vee
\Fchance
{\FAG}
{\phi_{j_m} \wedge \psi_{j_{m}}}
{\emptyset }
{\top }
\vee
\]
\[
\FBV_{j \in \FPI^{*}_{\FL}-\FPI^{*}_{m}} \Fchp{j}
\]

%%%%
\item

From (5) and (6), by application of the rule $\Frule_1$:
\[
\vdash_\FCCSRp
\FBV_{j^{*}\in \FPI^{*}_{m-1}} \Fchance
{\FAG}
{\phi_{j^{*}} \wedge \psi_{j^{*}} }
{\emptyset }
{\top }
\vee
\big(\Fchance
{\FAG}
{\phi_{j_m} \wedge \psi_{j_{m}} )}
{\emptyset }
{\top }
\wedge \Fchance
{\emptyset }
{\phi_{j_{m}}}
{\emptyset}
{\top}\big)
\vee
\]
\[
\FBV_{j \in \FPI^{*}_{\FL}-\FPI^{*}_{m}}\Fchp{j}
\]

%%%%
\item

Note $\FAA_{j_{m}}\cup \FBB_{j_{m}}=\FAG$. By application of the rule $\Frule_5$:
\[
\vdash_\FCCSRp
\FBV_{j^{*}\in \FPI^{*}_{m-1}} \Fchance
{\FAG }
{\phi_{j^{*}} \wedge \psi_{j^{*}} }
{\emptyset }
{\top }
\vee
\Fchance
{\FAA_{j_{m}}}
{\phi_{j_m}}
{\FBB_{j_{m}}}
{\psi_{j_m} }
\vee
\]
\[
\FBV_{j \in \FPI^{*}_{\FL}-\FPI^{*}_{m}}\Fchp{j}
\]

%%%%
\item

Note $j_{m} \notin \FPI^{*}_{m-1}$ and $\FPI^{*}_{m}=\FPI^{*}_{m-1} \cup \{j_{m}\}$. It can be verfied $\FPI^*_\FL- \FPI^*_{m-1} = (\FPI^{*}_\FL - \FPI^{*}_{m}) \cup \{j_{m}\}$. Then:
\[
\vdash_\FCCSRp
\FBV_{j^{*}\in \FPI^{*}_{m-1}} \Fchance
{\FAG }
{\phi_{j^{*}} \wedge \psi_{j^{*}} }
{\emptyset }
{\top}
\vee
\FBV_{j \in \FPI^{*}_{\FL}-\FPI^{*}_{m-1}}\Fchp{j}
\]
\end{enumerate}

Therefore, for every $k < \FL$:
\[
\vdash_\FCCSRp
\FBV_{j^{*}\in \FPI^{*}_{k}}
\Fchance
{\FAG}
{\phi_{j^{*}} \wedge \psi_{j^{*}}}
{\emptyset}
{\top}
\vee
\FBV_{j \in \FPI^{*}_{\FL} - \FPI^{*}_{k}} \Fchp{j}
\]

Specially:
\[
\vdash_\FCCSRp
\FBV_{j^{*}\in \FPI^{*}_{0}} \Fchance
{\FAG}
{\phi_{j^{*}} \wedge \psi_{j^{*}} )}
{\emptyset }
{\top}
\vee
\FBV_{j \in \FPI^{*}_{\FL}-\FPI^{*}_{0}}\Fchp{j}
\]

Note $\FAA_{j^{*}}=\FAG $ for all $j^{*}\in \FPI^{*}_{0}$. Then:
\[
\vdash_\FCCSRp
\FBV_{j^{*} \in \FPI^{*}_{0}} \Fchance
{\FAA_{j^{*}}}
{\phi_{j^{*}} \wedge \psi_{j^{*}} )}
{\emptyset }
{\top}
\vee
\FBV_{j \in \FPI^{*}_{\FL}-\FPI^{*}_{0}}\Fchp{j}
\]

By repeated applications of the rule $\Frule_4$:
\[
\vdash_\FCCSRp
\FBV_{j^{*}\in \FPI^{*}_{0}}\Fchp{j^{*}}
\vee
\FBV_{j \in \FPI^{*}_{\FL}-\FPI^{*}_{0}}\Fchp{j}
\]

That is:
\[
\vdash_\FCCSRp
\FBV_{j^{*} \in \FPI^{*}_{\FL}} \Fchp{j}
\]

By application of the rule $\Frule_1$:
\[
\vdash_\FCCSRp \gamma \vee \FBV_{j\in \FPI }\Fchp{j}
\]

\end{proof}

%%%%%%%%%%%%%%%%%%%%%%%%%%%
%%%%%%%%%%%%%%%%%%%%%%%%%%%
\subsection{Completeness of $\FCCSRp$ by induction}

%%%%%%%%%%%%%%%%%%%%%%%%
%%%%%%%%%%%%%%%%%%%%%%%%
\begin{theorem}[Completeness of $\FCCSRp$]

The axiomatic system for $\FCCSRp$ given in Definition \ref{definition:An axiomatic system for FCCSRp} is complete with respect to the set of valid formulas in $\Phi_{\FCCSRp}$.

\end{theorem}

%%%%%%%%%%%%%%%%%%%%%%%%
%%%%%%%%%%%%%%%%%%%%%%%%
\begin{proof}
~

Let $\phi$ be a formula in $\Phi_{\FCCSRp}$. Assume $\models \phi$. We want to show $\vdash_\FCCSRp \phi$. We put an induction on the modal depth $n$ of $\phi$.

Assume $n = 0$. Then, $\phi$ is a propositional tautology. Then, $\vdash_\FCCSRp \phi$.

Assume $n > 0$.

By Lemma \ref{lemma:normal-form CCSRp}, the normal form lemma for $\FCCSRp$, there is $\phi' \in \Phi_{\FCCSRp}$ such that (1) $\models \phi$ if and only if $\models \phi'$, (2) $\vdash_\FCCSRp \phi$ if and only if $\vdash_{\FCCSRp} \phi'$, (3) $\phi$ and $\phi'$ have the same modal depth, and (4) $\phi'$ is in the form of $\FSD_0 \land \dots \land \FSD_k$, where every $\FSD_i$ is a standard disjunction for $\FCCSRp$.

Then $\models \phi'$. Let $i \leq k$. Then $\models \FSD_i$. It suffices to show $\vdash_\FCCSRp \FSD_i$.

Assume the modal depth of $\FSD_i$ is less than $n$. By the inductive hypothesis, $\vdash_\FCCSRp \FSD_i$.

Assume the modal depth of $\FSD_i$ is $n$. Let $\FSD_i = \gamma \vee \FBV_{j \in \FPI} \Fchp{j}$.

By Lemma \ref{lemma:Downward validity for FCCSRp}, the second downward validity lemma for $\FCCSRp$, the second validity-reduction condition of $\FSD_i$ is met; that is, one of the following conditions holds:
\begin{enumerate}[label=(\alph*),leftmargin=3.33em]
\item $\vDash \gamma$;
\item there is a semantically well-arranged sequence of subsets of $\FPI$; that is, there is a sequence $(\FPI^{*}_{k})_{0 \leq k \leq \FL}$ of subsets of $\FPI$ for some $\FL \in \mathbb{N}$ such that:
\begin{itemize}
\item $\FPI^*_0 = \FPIf$;
\item for every $k < \FL$, $\FPIf \subseteq \FPI^{*}_k \subseteq \FPIfs$;
\item for every $k < \FL$, $\FPI^*_{k+1}$ is a semantic expansion of $\FPI^*_{k}$;
\item $\vDash \FBV_{j \in \FPI^*_\FL} (\phi_j \land \psi_j)$.
\end{itemize}

\end{enumerate}

By the inductive hypothesis, the derivability-reduction condition of $\FSD_i$ is met; that is, one of the following conditions holds:
\begin{enumerate}[label=(\alph*),leftmargin=3.33em]
\item $\vdash_\FCCSRp \gamma$;
\item there is a syntactically well-arranged sequence of subsets of $\FPI$; that is, there is a sequence $(\FPI^{*}_{k})_{0 \leq k \leq \FL}$ of subsets of $\FPI$ for some $\FL \in \mathbb{N}$ such that:
\begin{itemize}
\item $\FPI^*_0 = \FPIf$;
\item for every $k < \FL$, $\FPIf \subseteq \FPI^{*}_k \subseteq \FPIfs$;
\item for every $k < \FL$, $\FPI^*_{k+1}$ is a syntactical expansion of $\FPI^*_{k}$;
\item $\vdash_\FCCSRp \FBV_{j \in \FPI^*_\FL} (\phi_j \land \psi_j)$.
\end{itemize}

\end{enumerate}

By Lemma \ref{lemma:Upward derivability for FCCSRp}, the upward derivability lemma for $\FCCSRp$, $\vdash_\FCCSRp \FSD_i$.

\end{proof}

%%%%%%%%%%%%%%%%%%%%%%%%
%%%%%%%%%%%%%%%%%%%%%%%%
\section{Concluding remarks}
\label{section:Concluding remarks}

Previously, by the same approach, we showed the completeness of the liability and ability fragments of the Logic for Cooperating Conditional Strategic Reasoning $\FCCSR$.
The key ingredients of the approach include standard disjunctions, the validity-reduction condition of standard disjunctions, abstract game forms and their realization, and the derivability-reduction condition of standard disjunctions.
The approach is general, and after some adaptation, it can be applied to show the completeness of some other strategic logics. For example, by the same approach, we have proved the completeness of eight coalition logics~\cite{li_completeness_2024}, including Coalition Logic proposed by Pauly \cite{pauly_modal_2002}.

%%%%%%%%%%%%%%%%%%%%%%%%
%%%%%%%%%%%%%%%%%%%%%%%%
\subsection*{Acknowledgments}

We want to thank Valentin Goranko for his kind and great help with this project.
Thanks also go to the audience of a seminar and a workshop at Beijing Normal University.
This research was supported by the National Social Science Foundation of China (No. 19BZX137).

\bibliographystyle{alpha}
\bibliography{Strategy-reasoning}

\end{document}